\documentclass[a4paper,10pt]{article}

\usepackage{fullpage}
\usepackage{setspace}
\usepackage{slashbox}
\usepackage[latin1]{inputenc} % Italian accents are permitted
\usepackage{natbib}
\usepackage{amsmath}
\usepackage{amsfonts}
\usepackage{amssymb}
\usepackage{amsthm}
\usepackage{ctable}
\usepackage{subfigure}
\usepackage{graphicx}
\usepackage{cancel}
\usepackage{array}
\usepackage{color}
\usepackage{colortbl} 
\usepackage{multirow} % consente multi righe nelle tabelle
\usepackage{subfigure} % generates subfigures and subtables
\usepackage{latexsym}
\usepackage{textcomp}
\usepackage{tipa}

\newtheorem{lem}{Proposition} %[thm]
\begin{document}

\title{Flexible Mixture Modeling with the\\ Polynomial Gaussian Cluster-Weighted Model} 

\author{
Antonio Punzo\thanks{Dipartimento di Economia e Impresa - Universit\`{a} di Catania (Italy), e.mail: \texttt{antonio.punzo@unict.it}} 
%Luca Bagnato\thanks{Dipartimento di Metodi Quantitativi per le Scienze Economiche e Aziendali - Universit\`{a} di Milano-Bicocca (Italy), e.mail: \texttt{luca.bagnato@unimib.it}}\quad  
%Lucio De Capitani\thanks\thanks{Dipartimento di Metodi Quantitativi per le Scienze Economiche e Aziendali - Universit\`{a} di Milano-Bicocca (Italy), e.mail: \texttt{lucio.decapitani1@unimib.it}}\quad  
%Antonio Punzo
}

\date{}

\maketitle

%\vspace{0.8cm}

%\noindent \textbf{Address to which proofs should be sent:}
%
%%\vspace{0.5cm}
%\noindent \textit{Name}: Antonio Punzo\\
%\textit{Title}: Researcher \\
%\textit{Institution}: Facolt\`{a} di Economia, Dipartimento di Impresa, Culture e Societ\`{a}, Universit\`{a} di Catania (Italy) \\
%\textit{Address}: Corso Italia, 55, 95129 Catania, Italy \\
%\textit{Telephone number}: +39 095-7537640\\
%\textit{Fax}: +39 095-7537610\\
%\textit{e-mail}: \texttt{antonio.punzo@unict.it} 

\begin{abstract}

In the mixture modeling frame, this paper presents the polynomial Gaussian cluster-weighted model (CWM).
It extends the linear Gaussian CWM, for bivariate data, in a twofold way.
Firstly, it allows for possible nonlinear dependencies in the mixture components by considering a polynomial regression.
Secondly, it is not restricted to be used for model-based clustering only being contextualized in the most general model-based classification framework.
Maximum likelihood parameter estimates are derived using the EM algorithm and model selection is carried out using the Bayesian information criterion (BIC) and the integrated completed likelihood (ICL).
The paper also investigates the conditions under which the posterior probabilities of component-membership from a polynomial Gaussian CWM coincide with those of other well-established mixture-models which are related to it.
With respect to these models, the polynomial Gaussian CWM has shown to give excellent clustering and classification results when applied to the artificial and real data considered in the paper.

\end{abstract}


\section{Introduction}
\label{sec:introduction}

Finite mixture models are commonly employed in statistical modeling with two different purposes \citep[][pp.~2--3]{Titt:Smit:Mako:stat:1985}.
In indirect applications, they are used as semiparametric competitors of nonparametric density estimation techniques (see \citealt[][pp.~28--29]{Titt:Smit:Mako:stat:1985}, \citealt[][p.~8]{McLa:Peel:fini:2000} and \citealt{Esco:West:Baye:1995}).
On the other hand, in direct applications, finite mixture models are considered as a powerful device for clustering and classification by assuming that each mixture-component represents a group (or cluster)
in the original data (see \citealt{Fral:Raft:Howm:1998} and \citealt{McLa:Basf:mixt:1988}).
The areas of application of mixture models range from biology and medicine \citep[see][]{Schla:Medi:2009} to economics and marketing \citep[see][]{Wede:Kama:Mark:2001}.
An overview on mixture models is given in 
%\cite{EvHa:fini:1981}, \cite{Titt:Smit:Mako:stat:1985}, \cite{McLa:Basf:mixt:1988}, 
\cite{McLa:Peel:fini:2000} and \cite{Fruh:Fine:2006}.

This paper focuses both on direct and indirect applications.
The context of interest is represented by data arising from a real-valued bivariate random vector $\left(X,Y\right)'$ in which a functional dependence of $Y$ on $x$ is assumed for each mixture-component.
Note that, hereafter, all vectors are considered to be column vectors.
%If this functional relation is supposed to be linear, 
The linear cluster-weighted model (CWM; \citealt{Gers:Nonl:1997}) constitutes a natural choice, in the mixtures frame, when this functional relationship is supposed to be linear.
The (linear) CWM factorizes the joint density of $\left(X,Y\right)'$, in each mixture-component, into the product of the conditional density of $Y|x$ and the marginal density of $X$.
A Gaussian distribution is usually used for both of them leading to the (linear) Gaussian CWM 
(see \citealp{Ingr:Mino:Vitt:Loca:2011,Ingr:Mino:Punz:Mode:2012} for details on this model and for an extension to the $t$ distribution).
%\citep[see][for details on this model and for an extension to the $t$ distribution]{Ingr:Mino:Vitt:Loca:2011}.  
Generally, regardless from the density shape, CWMs take simultaneously into account the potential of finite mixtures of regressions \citep[see][Chapter~8]{Fruh:Fine:2006} and of finite mixtures of distributions (see, \citealt{Titt:Smit:Mako:stat:1985} and \citealt{McLa:Peel:fini:2000}); the idea of the former approach is adopted to model the conditional density of $Y|x$, while the principle of the latter is used to model both the joint density of $\left(X,Y\right)'$ and the marginal density of $X$.
Unfortunately, if the component-dependence structure of $Y$ on $x$ is different from the linear one, the linear CWM is not able to capture it.

To solve this problem, the present paper illustrates the polynomial Gaussian CWM.
It generalizes the linear Gaussian CWM by using a polynomial model for the dependence of $Y$ on $x$ in each mixture-component.
Regarding the comparison between a polynomial Gaussian CWM and a finite mixture of polynomial Gaussian regressions, two related aspects need to be preliminarily highlighted: from an indirect point of view, in fact, they are not comparable since the latter is not conceived to model the joint density of $\left(X,Y\right)'$, but only the conditional density of $Y|x$, and this difference, from a direct point of view, implicitly affects the criterion (vertical distance for finite mixture of regressions) used for clustering and classification. 
%From a direct point of view, the difference in the considered densities implicitly affects clustering and classification.
%, on which clustering and classification of the two models are based, implicitly affects  
%Furthermore, the underlying  and these densities are used
%Secondly, these are the densities which, from a direct point of view, are used for clustering and classification. %from a finite mixture of polynomial Gaussian regressions
%obtained clustering and classification differs since the densities on which are based differ too.
With respect to the latter point, in the paper we will show both artificial and real data on which the polynomial Gaussian CWM outperforms finite mixtures of polynomial Gaussian regressions.
Also,
%The advantages of the polynomial Gaussian CWM are also in terms of indirect applications.
%In particular, 
most of the existing literature and applications for bivariate mixture models are based on Gaussian components. 
Mixtures of $t$ distributions have been considered as robust alternatives (see \citealt{Peel:McLa:2000} and \citealt{Gres:Ingr:Cons:2010}).
These two models have been widely used for model-based clustering (\citealt{Cele:Gova:Gaus:1995}, \citealt{Banf:Raft:mode:1993}, and \citealt{McLa:Peel:Robu:1998}) and model-based classification (\citealt{Dean:Murp:Down:Usin:2006} and \citealt{Andr:McNi:Sube:Mode:2011}).
However, Gaussian and $t$ models imply that the groups are elliptically contoured, which is rather restrictive. 
Furthermore, it is well known that non-elliptical subpopulations can be approximated quite well by a mixture of several basic densities like the Gaussian one (\citealt{Fral:Raft:Howm:1998}, \citealt{Dasg:Raft:Dete:1998}, \citealt{McLa:Peel:fini:2000} p.~1 and \citealt{Titt:Smit:Mako:stat:1985} p.~24). 
While this can be very helpful for modeling purposes, it can be misleading when dealing with clustering
applications since one group may be represented by more than one component just because it has, in fact, a non-elliptical density. 
A common approach to treat non-elliptical subpopulations consists in considering transformations so as to make the components as elliptical as possible and then fitting symmetric (usually Gaussian) mixtures (\citealt{Guti:Carr:Wang:Lee:Tayl:Anal:1995} and \citealt{Lo:Brin:Gott:Cyto:2008}).
Although such a treatment is very convenient to use, the achievement of joint ellipticality is rarely satisfied and the transformed variables become more difficult to interpret.
Instead of applying transformations, there has been a growing interest in proposing finite mixture models where the component densities are non-elliptical (above all in terms of skewness) so as to represent correctly the non-elliptical subpopulations (\citealt{Karl:Sant:Mode:2009} and \citealt{Lin:Maxi:2009,Lin:Robu:2010}). While such models can be used to create finite mixture models and provide alternative shapes for the derived clusters, they have certain computational difficulties.
On the contrary the polynomial Gaussian CWM, which allows the bivariate density components to be flexible enough by allowing the polynomial degree to increase, is easily applicable, has much simpler expressions for estimation purposes, and generates easily interpreted clusters.

The paper is organized as follows.
In Section~\ref{sec:Model Definition} the polynomial Gaussian CWM is presented and, in Section~\ref{sec:Modeling framework}, its use is contextualized in the presence of possible labeled observations among the available ones.
In Section~\ref{sec:EM}, maximum likelihood estimation of the model parameters is approached by considering, and detailing, the EM-algorithm.
With this regard, computational details are given in Section~\ref{subsec:Computational issues} while methods to select the number of components and the polynomial degree are given in Section~\ref{sec:Model selection and performance evaluation}.
Some theoretical notes about the relation, from a direct point of view, with the classification provided by other models in the mixture frame, are given in Section~\ref{subsec:preliminary considerations}.
Artificial and real data are considered in Section~\ref{subsec:Artificial data} and Section~\ref{subsec:Real data}, and the paper closes, with discussion and suggestions for further work, in Section~\ref{sec:conclusions}. 

\section{Model Definition}
\label{sec:Model Definition}

Let 
\begin{equation}
p\left(x,y;\boldsymbol{\psi}\right)=\sum_{j=1}^k \pi_j f\left(x,y;\boldsymbol{\vartheta}_j\right) 
\label{eq:mixture model}
\end{equation}
be the finite mixture of distributions, with $k$ components, used to estimate the joint density of $\left(X,Y\right)'$. 
%Let $p\left(x,y\right)$ be the joint density of $\left(X,Y\right)'$.
%It can be modelized by the finite mixture model 
%\begin{equation}
%p\left(x,y;\boldsymbol{\psi}\right)=\sum_{j=1}^k \pi_j f\left(x,y;\boldsymbol{\vartheta}_j\right), 
%\label{eq:mixture model}
%\end{equation}
In \eqref{eq:mixture model}, $f\left(\cdot;\boldsymbol{\vartheta}_j\right)$ is the parametric (with respect to the vector $\boldsymbol{\vartheta}_j$) density associated to the $j$th component, $\pi_j$ is the weight of the $j$th component, with $\pi_j>0$ and $\sum_{j=1}^k\pi_j=1$, and $\boldsymbol{\psi}=\left(\boldsymbol{\pi}',\boldsymbol{\vartheta}'\right)'$, with $\boldsymbol{\pi}=\left(\pi_1,\ldots,\pi_{k-1}\right)'$ and $\boldsymbol{\vartheta}=\left(\boldsymbol{\vartheta}_1,\ldots,\boldsymbol{\vartheta}_k\right)'$, contains all the unknown parameters in the mixture.
As usual, model \eqref{eq:mixture model} implicitly assumes that the component densities should all belong to the same parametric family.
%In the following we will refer to $f\left(x,y;\boldsymbol{\vartheta}_j\right)$ as mixture components, and the groups in the data induced by these components as clusters.

Now, suppose that for each $j$ the functional dependence of $Y$ on $x$ can be modeled as 
\begin{equation}
Y=\mu_j\left(x\right)+\varepsilon_j, 
\label{eq:regression model}
\end{equation}
where $\mu_j\left(x\right)=E\left(Y|X=x,j\right)$ is the regression function and $\varepsilon_j$ is the error variable having a Gaussian distribution with zero mean and a finite constant variance $\sigma_{\varepsilon_j}^2$, hereafter simply denoted by $\varepsilon_j\sim N\left(0,\sigma_{\varepsilon_j}^2\right)$.
Thus, for each $j$, $Y|x\sim N\left(\mu_j\left(x\right),\sigma_{\varepsilon_j}^2\right)$.   
In the parametric paradigm, the polynomial regression function 
\begin{equation}
\mu_j\left(x\right)=\mu_r\left(x;\boldsymbol{\beta}_j\right)=\sum_{l=0}^r\beta_{lj}x^l=\boldsymbol{\beta}_j'\boldsymbol{x}
\label{eq:polynomial regression function}
\end{equation}
represents a very flexible way to model the functional dependence in each component.
In \eqref{eq:polynomial regression function}, $\boldsymbol{\beta}_j=\left(\beta_{0j},\beta_{1j},\ldots,\beta_{rj}\right)'$ is the $\left(r+1\right)$-dimensional vector of real parameters, $\boldsymbol{x}$ is the $\left(r+1\right)$-dimensional Vandermonde vector associated to $x$, while $r$ is an integer representing the polynomial order (degree) which is assumed to be fixed with respect to $j$.

The mixture model \eqref{eq:mixture model} becomes a polynomial Gaussian CWM when the $j$th component joint density is factorized as
\begin{equation}
f\left(x,y;\boldsymbol{\vartheta}_j\right)=\phi\left(y\left|x;\mu_r\left(x;\boldsymbol{\beta}_j\right),\sigma_{\varepsilon_j}^2\right.\right)\phi\left(x;\mu_{X|j},\sigma_{X|j}^2\right).
\label{eq:component density}
\end{equation}
In \eqref{eq:component density}, $\phi\left(\cdot\right)$ denotes a Gaussian density; this means that, for each $j$, $X\sim N\left(\mu_{X|j},\sigma_{X|j}^2\right)$.
Summarizing, the polynomial Gaussian CWM has equation
\begin{equation}
p\left(x,y;\boldsymbol{\psi}\right)=
\sum_{j=1}^k\pi_j\phi\left(y\left|x;\mu_r\left(x;\boldsymbol{\beta}_j\right),\sigma_{\varepsilon_j}^2\right.\right)\phi\left(x;\mu_{X|j},\sigma_{X|j}^2\right).
\label{eq:polynomial CWM}
\end{equation}
Note that the number of free parameters in \eqref{eq:polynomial CWM} is $\eta=kr+4k-1$. %\textbf{(k-1) + 2*k + k*(r+1)}
Moreover, if $r=1$ in equation \eqref{eq:polynomial regression function}, model \eqref{eq:polynomial CWM} corresponds to the linear Gaussian CWM widely analyzed in \cite{Ingr:Mino:Vitt:Loca:2011}. 
%a Gaussian density $\phi\left(\cdot\right)$ is also assumed for the marginal distribution of $X$. 
%is obtained by 
%By considering a Gaussian density $\phi\left(\cdot\right)$ also for the marginal distribution of $X$ in each mixture-component, the $j$th component joint density of the mixture model in \eqref{eq:mixture model} can be written as

% and $Y|x$, starting from \eqref{eq:mixture model}, 
%we obtain the polynomial Gaussian CWM of equation
%\begin{equation}
%p\left(x,y;\boldsymbol{\psi}\right)=
%\sum_{j=1}^k\pi_j\phi\left(y\left|x,j;\mu_r\left(x;\boldsymbol{\beta}_j\right),\sigma_{\varepsilon_j}\right.\right)\phi\left(x\left|j;\mu_{X|j},\sigma_{X|j}\right.\right),
%\label{eq:polynomial CWM}
%\end{equation}
%where $\mu_{X|j}$ and $\sigma_{X|j}$ represent mean and standard deviation of $X$, respectively, and where, with reference to \eqref{eq:mixture model}, $\boldsymbol{\vartheta}_j=\left\{\boldsymbol{\beta}_j,\sigma_{\varepsilon_j},\mu_{X|j},\sigma_{X|j}\right\}$. 

\section{Modeling framework}
\label{sec:Modeling framework}

%\textbf{Specificare un simbolo, tipo S, per il campione.}

As said in Section~\ref{sec:introduction}, the polynomial Gaussian CWM, being a mixture model, can be also used for direct applications, where the aim is to clusterize/classify observations which have unknown component memberships (the so-called unlabeled observations).
To embrace both clustering and classification purposes, we have chosen a very general scenario where there are $n$ observations $\left(x_1,y_1\right)',\ldots,\left(x_n,y_n\right)'$, $m$ of which are labeled.
As a special case, if $m=0$, we obtain the clustering scenario. 
Within the model-based classification framework, we use all the $n$ observations to estimate the parameters in \eqref{eq:polynomial CWM}; the fitted mixture model is so adopted to classify each of the $n-m$ unlabeled observations through the corresponding maximum \textit{a posteriori} probability (MAP). 
Drawing on \citet{Hosm:acom:1973}, \citet[][Section~4.3.3]{Titt:Smit:Mako:stat:1985} pointed out that knowing the label of just a small proportion of observations \textit{a priori} can lead to improved clustering performance.  

Notationally, let $\boldsymbol{z}_i$ be the $k$-dimensional component-label vector in which the $j$th element $z_{ij}$ is defined to be one or zero according to whether the mixture-component of origin of $\left(x_i,y_i\right)'$ is equal to $j$ or not, $j=1,\ldots,k$.
%Notationally, let $z_{ij}$ be the component membership indicator of the $i$th observation, where $z_{ij}=1$ if $\left(x_i,y_i\right)'$ belongs to component $j$ and $z_{ij}=0$ otherwise.
If the $i$th observation is labeled, denote with $\widetilde{\boldsymbol{z}}_i=\left(\widetilde{z}_{i1},\ldots,\widetilde{z}_{ik}\right)$ its component membership indicator. 
Then, arranging the data so that the first $m$ observations are labeled, the observed sample can be denoted by $\mathcal{S}=\left\{\mathcal{S}_l,\mathcal{S}_u\right\}$, where $\mathcal{S}_l=\left\{\left(x_1,y_1,\widetilde{\boldsymbol{z}}_1'\right)',\ldots,\left(x_m,y_m,\widetilde{\boldsymbol{z}}_m'\right)'\right\}$ is the sample of labeled observations while $\mathcal{S}_u=\left\{\left(x_{m+1},y_{m+1}\right)',\ldots,\left(x_n,y_n\right)'\right\}$ is the sample of unlabeled observations.
The completed-data sample can be so indicated by $\mathcal{S}_c=\left\{\mathcal{S}_l,\mathcal{S}_u^*\right\}$, where $\mathcal{S}_u^*=\left\{\left(x_{m+1},y_{m+1},\boldsymbol{z}_{m+1}'\right)',\ldots,\left(x_n,y_n,\boldsymbol{z}_n'\right)'\right\}$. 
Hence, the observed-data log-likelihood for the polynomial Gaussian CWM, when both $k$ and $r$ are supposed to be pre-assigned, can be written as
\begin{equation}
l\left(\boldsymbol{\psi}\right)=\sum_{i=1}^m\sum_{j=1}^k\widetilde{z}_{ij}\left[\ln\pi_j+\ln f\left(x_i,y_i;\boldsymbol{\vartheta}_j\right)\right]+\sum_{i=m+1}^n\ln\left[\sum_{j=1}^k\pi_jf\left(x_i,y_i;\boldsymbol{\vartheta}_j\right)\right],
\label{eq:observed-data log-likelihood}
\end{equation}
while the complete-data log-likelihood is
\begin{eqnarray}
l_c\left(\boldsymbol{\psi}\right)&=&\sum_{i=1}^m\sum_{j=1}^k\widetilde{z}_{ij}\left[\ln \pi_j+\ln f\left(x_i,y_i;\boldsymbol{\vartheta}_j\right)\right]+\sum_{i=m+1}^n\sum_{j=1}^kz_{ij}\left[\ln \pi_j+\ln f\left(x_i,y_i;\boldsymbol{\vartheta}_j\right)\right]\nonumber\\
&=&l_{1c}\left(\boldsymbol{\pi}\right)+l_{2c}\left(\boldsymbol{\kappa}\right)+l_{3c}\left(\boldsymbol{\xi}\right),
\label{eq:complete-data log-likelihood}
\end{eqnarray}
where $\boldsymbol{\kappa}=\left(\boldsymbol{\kappa}_1',\ldots,\boldsymbol{\kappa}_k'\right)'$ with $\boldsymbol{\kappa}_j=\left(\mu_{X|j},\sigma_{X|j}^2\right)'$, $\boldsymbol{\xi}=\left(\boldsymbol{\xi}_1',\ldots,\boldsymbol{\xi}_k'\right)'$ with $\boldsymbol{\xi}_j=\left(\boldsymbol{\beta}_j,\sigma_{\varepsilon_j}^2\right)'$, and 
\begin{eqnarray}
l_{1c}\left(\boldsymbol{\pi}\right)&=&\sum_{i=1}^m\sum_{j=1}^k\widetilde{z}_{ij}\ln \pi_j+\sum_{i=m+1}^n\sum_{j=1}^kz_{ij}\ln \pi_j \label{eq:lC1}\\
l_{2c}\left(\boldsymbol{\kappa}\right)
&=&\frac{1}{2}\sum_{i=1}^m\sum_{j=1}^k\widetilde{z}_{ij}\left[-\ln\left(2\pi\right)-\ln\left(\sigma_{\varepsilon_j}^2\right)-\frac{\left(y_i-\boldsymbol{\beta}_j'\boldsymbol{x}_i\right)^2}{\sigma_{\varepsilon_j}^2}\right]+\nonumber\\
&&+\frac{1}{2}\sum_{i=m+1}^n\sum_{j=1}^kz_{ij}\left[-\ln\left(2\pi\right)-\ln\left(\sigma_{\varepsilon_j}^2\right)-\frac{\left(y_i-\boldsymbol{\beta}_j'\boldsymbol{x}_i\right)^2}{\sigma_{\varepsilon_j}^2}\right] \label{eq:lC2}\\
l_{3c}\left(\boldsymbol{\xi}\right) &=& \frac{1}{2}\sum_{i=1}^m\sum_{j=1}^k\widetilde{z}_{ij}\left[-\ln\left(2\pi\right)-\ln\left(\sigma_{X|j}^2\right)-\frac{\left(x_i-\mu_{X|j}\right)^2}{\sigma_{X|j}^2}\right]+\nonumber\\
&&+\frac{1}{2}\sum_{i=m+1}^n\sum_{j=1}^kz_{ij}\left[-\ln\left(2\pi\right)-\ln\left(\sigma_{X|j}^2\right)-\frac{\left(x_i-\mu_{X|j}\right)^2}{\sigma_{X|j}^2}\right]. \label{eq:lC3}
\end{eqnarray}

\section{The EM algorithm for maximum likelihood estimation}
\label{sec:EM}

The EM algorithm \citep{Demp:Lair:Rubi:Maxi:1977} can be used to maximize $l\left(\boldsymbol{\psi}\right)$ in order to find maximum likelihood (ML) estimates for the unknown parameters of the polynomial Gaussian CWM.
When both labeled and unlabeled data are used, the E and M steps of the algorithm can be detailed as follows. 

%\subsection{The EM algorithm}
%\label{subsec:The EM algorithm}

%Maximum likelihood parameter estimation is carried out using the expectation-maximization (EM) algorithm \citep{Demp:Lair:Rubi:Maxi:1977}, as will be described in Section~\ref{subsec:The EM algorithm}.
%Looking at \ref{eq:observed-data log-likelihood}, one may notice that setting $m=0$ gives the log-likelihood for the model-based clustering scenario. 
%In fact, model-based clustering can be viewed as a special case of model-based classification where all the observations are unlabeled.
%
%
%Maximum likelihood parameter estimation is carried out using the expectation-maximization (EM) algorithm \citep{Demp:Lair:Rubi:Maxi:1977}, as will be described in Section~\ref{subsec:The EM algorithm}.

\subsection[E-step]{E-step}
\label{subsec:E-step}

The E-step, on the $\left(q+1\right)$th iteration, requires the calculation of
\begin{equation}
Q\left(\boldsymbol{\psi};\boldsymbol{\psi}^{\left(q\right)}\right)=E_{\boldsymbol{\psi}^{\left(q\right)}}\left[l_c\left(\boldsymbol{\psi}\right)\left|\mathcal{S}_u^*\right.\right].
\label{eq:Q}
\end{equation}
As $l_c\left(\boldsymbol{\psi}\right)$ is linear in the unobservable data $z_{ij}$, the E-step -- on the $\left(q+1\right)$th iteration -- simply requires the calculation of the current conditional expectation of $Z_{ij}$ given the observed sample, where $Z_{ij}$ is the random variable corresponding to $z_{ij}$.
In particular, for $i=m+1,\ldots,n$ and $j=1,\ldots,k$, it follows that
\begin{eqnarray}
E_{\boldsymbol{\psi}^{\left(q\right)}}\left(Z_{ij}\left|\mathcal{S}_u^*\right.\right)&=&z_{ij}^{\left(q\right)}\nonumber\\
&=&\frac{\pi_j^{\left(q\right)} f\left(x_i,y_i;\boldsymbol{\vartheta}_j^{\left(q\right)}\right)}{p\left(x_i,y_i;\boldsymbol{\psi}^{\left(q\right)}\right)},
\label{eq:EZ}
\end{eqnarray}
which corresponds to the posterior probability that the unlabeled observation $\left(x_i,y_i\right)'$ belongs to the $j$th component of the mixture, using the current fit $\boldsymbol{\psi}^{\left(q\right)}$ for $\boldsymbol{\psi}$.
By substituting the values $z_{ij}$ in \eqref{eq:complete-data log-likelihood} with the values $z_{ij}^{\left(q\right)}$ obtained in \eqref{eq:EZ}, we have
\begin{equation}
Q\left(\boldsymbol{\psi};\boldsymbol{\psi}^{\left(q\right)}\right)=Q_1\left(\boldsymbol{\pi};\boldsymbol{\psi}^{\left(q\right)}\right)+Q_2\left(\boldsymbol{\kappa};\boldsymbol{\psi}^{\left(q\right)}\right)+Q_3\left(\boldsymbol{\xi};\boldsymbol{\psi}^{\left(q\right)}\right),
\label{eq:decomposition Q}
\end{equation} 
where
\begin{eqnarray}
Q_1\left(\boldsymbol{\pi};\boldsymbol{\psi}^{\left(q\right)}\right)&=&\sum_{i=1}^m\sum_{j=1}^k\widetilde{z}_{ij}\ln \pi_j+\sum_{i=m+1}^n\sum_{j=1}^kz_{ij}^{\left(q\right)}\ln \pi_j \label{eq:Q1}\\
Q_2\left(\boldsymbol{\kappa};\boldsymbol{\psi}^{\left(q\right)}\right)
&=&\frac{1}{2}\sum_{i=1}^m\sum_{j=1}^k\widetilde{z}_{ij}\left[-\ln\left(2\pi\right)-\ln\left(\sigma_{\varepsilon_j}^2\right)-\frac{\left(y_i-\boldsymbol{\beta}_j'\boldsymbol{x}_i\right)^2}{\sigma_{\varepsilon_j}^2}\right]+\nonumber\\
&&+\frac{1}{2}\sum_{i=m+1}^n\sum_{j=1}^kz_{ij}^{\left(q\right)}\left[-\ln\left(2\pi\right)-\ln\left(\sigma_{\varepsilon_j}^2\right)-\frac{\left(y_i-\boldsymbol{\beta}_j'\boldsymbol{x}_i\right)^2}{\sigma_{\varepsilon_j}^2}\right] \label{eq:Q2}\\
Q_3\left(\boldsymbol{\xi};\boldsymbol{\psi}^{\left(q\right)}\right) &=& \frac{1}{2}\sum_{i=1}^m\sum_{j=1}^k\widetilde{z}_{ij}\left[-\ln\left(2\pi\right)-\ln\left(\sigma_{X|j}^2\right)-\frac{\left(x_i-\mu_{X|j}\right)^2}{\sigma_{X|j}^2}\right]+\nonumber\\
&&+\frac{1}{2}\sum_{i=m+1}^n\sum_{j=1}^kz_{ij}^{\left(q\right)}\left[-\ln\left(2\pi\right)-\ln\left(\sigma_{X|j}^2\right)-\frac{\left(x_i-\mu_{X|j}\right)^2}{\sigma_{X|j}^2}\right]. \label{eq:Q3}
\end{eqnarray}

\subsection[M-step]{M-step}
\label{subsec:M-step}

On the M-step, at the $\left(q+1\right)$th iteration, it follows from \eqref{eq:decomposition Q} that $\boldsymbol{\pi}^{\left(q+1\right)}$, $\boldsymbol{\kappa}^{\left(q+1\right)}$ and $\boldsymbol{\xi}^{\left(q+1\right)}$ can be computed independently of each other, by separate maximization of \eqref{eq:Q1}, \eqref{eq:Q2} and \eqref{eq:Q3}, respectively.
Here, it is important to note that the solutions exist in closed form.

Regarding the mixture weights, maximization of $Q_1\left(\boldsymbol{\pi};\boldsymbol{\psi}^{\left(q\right)}\right)$ with respect to $\boldsymbol{\pi}$, subject to the constraints on those parameters, is obtained by maximizing the augmented function 
\begin{equation}
\sum_{i=1}^m\sum_{j=1}^k\widetilde{z}_{ij}\ln \pi_j+\sum_{i=m+1}^n\sum_{j=1}^kz_{ij}^{\left(q\right)}\ln \pi_j-\lambda\left(\sum_{j=1}^k\pi_j-1\right),
\label{eq:Lagrange}
\end{equation}
where $\lambda$ is a Lagrangian multiplier.
Setting the derivative of equation \eqref{eq:Lagrange} with respect to $\pi_j$ equal to zero and solving for $\pi_j$ yields to
\begin{equation}
\pi_j^{\left(q+1\right)}=\displaystyle\frac{\displaystyle\sum_{i=1}^m\widetilde{z}_{ij}+\sum_{i=m+1}^nz_{ij}^{\left(q\right)}}{n}.
\label{eq:updated weights}
\end{equation}  

With reference to the updated estimates of $\boldsymbol{\kappa}_j$, $j=1,\ldots,k$, maximization of $Q_2\left(\boldsymbol{\kappa};\boldsymbol{\psi}^{\left(q\right)}\right)$ leads to
\begin{eqnarray*}
\mu_{X|j}^{\left(q+1\right)} 
&=&
\displaystyle\frac{\displaystyle\sum_{i=1}^m\widetilde{z}_{ij}x_i+\sum_{i=m+1}^nz_{ij}^{\left(q\right)}x_i}{\displaystyle\sum_{i=1}^m\widetilde{z}_{ij}+\sum_{i=m+1}^nz_{ij}^{\left(q\right)}}\label{eq:updated mean vector for X} \\
\sigma_{X|j}^{\left(q+1\right)} 
&=& \left(\displaystyle\frac{\displaystyle\sum_{i=1}^m\widetilde{z}_{ij}\left(x_i-\mu_{X|j}^{\left(q+1\right)}\right)^2+\sum_{i=m+1}^nz_{ij}^{\left(q\right)}\left(x_i-\mu_{X|j}^{\left(q+1\right)}\right)^2}{\displaystyle\sum_{i=1}^m\widetilde{z}_{ij}+\sum_{i=m+1}^nz_{ij}^{\left(q\right)}}\right)^{1/2}.
\label{eq:updated covariance matrix for X}
\end{eqnarray*}

Finally, regarding the update estimates of $\boldsymbol{\xi}_j$, $j=1,\ldots,k$, maximization of $Q_3\left(\boldsymbol{\xi};\boldsymbol{\psi}^{\left(q\right)}\right)$, after some algebra, yields to
\begin{eqnarray*}
\beta_{lj}^{\left(q+1\right)}&=&\left(\frac{\displaystyle\sum_{i=1}^m\widetilde{z}_{ij}\text{\textsubtilde{$\boldsymbol{x}_i$}}\text{\textsubtilde{$\boldsymbol{x}_i$}}'+\sum_{i=m+1}^nz_{ij}^{\left(q\right)}\text{\textsubtilde{$\boldsymbol{x}_i$}}\text{\textsubtilde{$\boldsymbol{x}_i$}}' }{\displaystyle\sum_{i=1}^m\widetilde{z}_{ij}+\sum_{i=m+1}^nz_{ij}^{\left(q\right)}}-\frac{\displaystyle\sum_{i=1}^m\widetilde{z}_{ij}\text{\textsubtilde{$\boldsymbol{x}_i$}}+\sum_{i=m+1}^n z_{ij}^{\left(q\right)}\text{\textsubtilde{$\boldsymbol{x}_i$}}
}{\displaystyle\sum_{i=1}^m\widetilde{z}_{ij}+\sum_{i=m+1}^nz_{ij}^{\left(q\right)}
}\frac{\displaystyle\sum_{i=1}^m\widetilde{z}_{ij}\text{\textsubtilde{$\boldsymbol{x}_i$}}'+\sum_{i=m+1}^n z_{ij}^{\left(q\right)}\text{\textsubtilde{$\boldsymbol{x}_i$}}'}{\displaystyle\sum_{i=1}^m\widetilde{z}_{ij}
+\sum_{i=m+1}^nz_{ij}^{\left(q\right)}}\right)^{-1}
\cdot\nonumber\\
&&\cdot \left(
\frac{\displaystyle\sum_{i=1}^m\widetilde{z}_{ij}^{\left(q\right)}y_i\text{\textsubtilde{$\boldsymbol{x}_i$}}+\sum_{i=m+1}^nz_{ij}^{\left(q\right)}y_i\text{\textsubtilde{$\boldsymbol{x}_i$}}}{\displaystyle\sum_{i=1}^m\widetilde{z}_{ij}+\sum_{i=m+1}^nz_{ij}^{\left(q\right)}}
-
\frac{\displaystyle\sum_{i=1}^m\widetilde{z}_{ij}y_i+\sum_{i=m+1}^nz_{ij}^{\left(q\right)}y_i}{\displaystyle\sum_{i=1}^m\widetilde{z}_{ij}
+\sum_{i=m+1}^nz_{ij}^{\left(q\right)}}\frac{\displaystyle\sum_{i=1}^m\widetilde{z}_{ij}\text{\textsubtilde{$\boldsymbol{x}_i$}}+\sum_{i=m+1}^nz_{ij}^{\left(q\right)}\text{\textsubtilde{$\boldsymbol{x}_i$}}}{\displaystyle\sum_{i=1}^m\widetilde{z}_{ij}
+\sum_{i=m+1}^nz_{ij}^{\left(q\right)}}\right)
\label{eq:updated beta1},\quad l=1,\ldots,r\\
\beta_{0j}^{\left(q+1\right)}&=&
\displaystyle\frac{\displaystyle\sum_{i=1}^m\widetilde{z}_{ij}y_i+\sum_{i=m+1}^nz_{ij}^{\left(q\right)}y_i}{\displaystyle\sum_{i=1}^m\widetilde{z}_{ij}+\sum_{i=m+1}^nz_{ij}^{\left(q\right)}}-\sum_{l=1}^r\beta_{lj}^{\left(q+1\right)}\displaystyle\frac{\displaystyle\sum_{i=1}^m\widetilde{z}_{ij}x^l_i+\sum_{i=m+1}^nz_{ij}^{\left(q\right)}x^l_i}{\displaystyle\sum_{i=1}^m\widetilde{z}_{ij}+\sum_{i=m+1}^nz_{ij}^{\left(q\right)}}
\label{eq:updated beta0}\\
\sigma^{\left(q+1\right)}_{\varepsilon_j}&=&
\left(\displaystyle\frac{\displaystyle\sum_{i=1}^m\widetilde{z}_{ij}\left(y_i-\boldsymbol{x}_i'\boldsymbol{\beta}_j^{\left(q+1\right)}\right)^2+\sum_{i=m+1}^nz_{ij}^{\left(q\right)}\left(y_i-\boldsymbol{x}_i'\boldsymbol{\beta}_j^{\left(q+1\right)}\right)^2}{\displaystyle\sum_{i=1}^m\widetilde{z}_{ij}+\sum_{i=m+1}^nz_{ij}^{\left(q\right)}}\right)^{1/2},\label{eq:updated conditioned variances}
\end{eqnarray*}
where the $r$-dimensional vector \textsubtilde{$\boldsymbol{x}_i$} is obtained from the Vandermonde vector $\boldsymbol{x}_i$ by deleting its first element.

\subsection{Some considerations}
\label{subsec:Some considerations}

In the following, the estimates obtained with the EM algorithm will be indicated with a hat.
Thus, for example, $\widehat{\boldsymbol{\psi}}$ and $\widehat{z}_{ij}$ will denote, respectively, the estimates of $\boldsymbol{\psi}$ and $z_{ij}$, $i=m+1,\ldots,n$ and $j=1,\ldots,k$.  
As said before, the fitted mixture model can be used to classify the $n-m$ unlabeled observations via the MAP classification induced by
\begin{equation}
\text{MAP}\left(\widehat{z}_{ij}\right)=\left\{
\begin{array}{ccl}
1 && \text{if $\max_h\left\{\widehat{z}_{ih}\right\}$ occurs at component $j$}\\
0 && \text{otherwise}\\
\end{array}
\right.
\label{eq:MAP}
\end{equation}
$i=m+1,\ldots,n$ and $j=1,\ldots,k$.
Note that the MAP classification is used in the analyses of Section~\ref{sec:illustrative examples}.

As an alternative to the EM algorithm, one could adopt the well-known Classification EM \citep[CEM;][]{Cele:Gova:Acla:1992} algorithm, although it maximizes $l_c\left(\boldsymbol{\psi}\right)$.
The CEM algorithm is almost identical to the EM algorithm, except for a further step, named C step, considered between the standard E and M ones.
In particular, in the C step the $z_{ij}^{\left(q\right)}$ in \eqref{eq:EZ} are substituted with $\text{MAP}\left(z_{ij}^{\left(q\right)}\right)$, and the obtained partition is used in the M-step.  
%In particular, after the E-step, a partition is obtained by assigning each unlabeled observation to the component in correspondence to the MAP.

% and the estimates obtained for $z_{ij}$, $i=m+1,\ldots,n$, directly classify the $n-m$ unlabeled observations.
%Alternatively, we can use the CEM algorithm \citep{Cele:Gova:Acla:1992} to maximize $l_c\left(\boldsymbol{\psi}\right)$ and the estimates obtained for $z_{ij}$, $i=m+1,\ldots,n$, directly classify the $m$ unlabeled observations.

%In detail, once the expected value of the unknown labels is computed according to \eqref{eq:EZ}, in the C step  

%In the following, we will describe these algorithms when both labeled and unlabeled data are used.
%The tractation will focus on the case in which both $k$ and $r$ are fixed.

\section{Computational issues}
\label{subsec:Computational issues}

Code for the EM algorithm (as well as its CEM variant) described in Section~\ref{sec:EM} was written in the \texttt{R} computing environment \citep{R}.

\subsection{EM initialization}
\label{subsec:EM initialization}

Before running the EM algorithm, the choice of the starting values constitutes an important issue.
The standard initialization consists in selecting a value for $\boldsymbol{\psi}^{\left(0\right)}$.
An alternative approach \citep[see][p.~54]{McLa:Peel:fini:2000}, more natural in the modeling frame described in Section~\ref{sec:Modeling framework}, is to perform the first E-step by specifying, in equation \eqref{eq:EZ}, the values of $\boldsymbol{z}_i^{\left(0\right)}$, $i=m+1,\ldots,n$, for the unlabeled observations.
Among the possible initialization strategies (see \citealt{Bier:Cele:Gova:Choo:2003} and \citealt{Karl:Xeka:Choo:2003} for details) -- according to the \texttt{R}-package \texttt{flexmix} (\citealt{Leis:Flex:2004} and \citealt{Grun:Leis:Flex:2008}) which allows to estimate finite mixtures of polynomial Gaussian regressions -- a random initialization is repeated $t$ times from different random positions and the solution maximizing the observed-data log-likelihood 
%(or the complete-data log-likelihood for the CEM) 
among these $t$ runs is selected.
In each run, the $n-m$ vectors $\boldsymbol{z}_i^{\left(0\right)}$ are randomly drawn from a multinomial distribution with probabilities $\left(1/k,\ldots,1/k\right)$.
%, and this scheme holds for both EM and CEM.

\subsection{Convergence criterion}
\label{subsec:Convergence criterion}

%To reduce the effect of stopping the algorithm too early, a rather strict stopping criterion for the iterations was used.
%In particular, iterations were stopped when the change in the log-likelihood -- which represents an indication of lack of progress \citep{Lind:Bate:Newt:1988} -- was smaller than $10^{-12}$; this is one of the possibilities considered in \citet{Karl:Xeka:Choo:2003}; see also \citet{Fral:Raft:Howm:1998}.

The Aitken acceleration procedure \citep{Aitk:OnBe:1926} is used to estimate the asymptotic maximum of the log-likelihood at each iteration of the EM algorithm. 
Based on this estimate, a decision can be made regarding whether or not the algorithm has reached convergence;
that is, whether or not the log-likelihood is sufficiently close to its estimated asymptotic value. 
The Aitken acceleration at iteration $k$ is given by
\begin{displaymath}
	a^{\left(k\right)}=\frac{l^{\left(k+1\right)}-l^{\left(k\right)}}{l^{\left(k\right)}-l^{\left(k-1\right)}},
\end{displaymath}
where $l^{\left(k+1\right)}$, $l^{\left(k\right)}$, and $l^{\left(k-1\right)}$ are the log-likelihood values from iterations $k+1$, $k$, and $k-1$, respectively. 
Then, the asymptotic estimate of the log-likelihood at iteration $k + 1$ \citep{Bohn:Diet:Scha:Schl:Lind:TheD:1994} is given by
\begin{displaymath}	l_{\infty}^{\left(k+1\right)}=l^{\left(k\right)}+\frac{1}{1-a^{\left(k\right)}}\left(l^{\left(k+1\right)}-l^{\left(k\right)}\right).
\end{displaymath}
In the analyses in Section~\ref{sec:illustrative examples}, we follow \citet{McNi:Mode:2010} and stop our algorithms when $l_{\infty}^{\left(k+1\right)}-l^{\left(k\right)}<\epsilon$, with $\epsilon=0.05$.

\subsection{Standard errors of the estimates}
\label{subsec:Estimating the degrees of freedom}

Once the EM algorithm is run, the covariance matrix of the estimated parameters $\widehat{\boldsymbol{\psi}}$ is determined by using the inverted negative Hessian matrix, as computed by the general purpose optimizer \verb|optim| in the \texttt{R}-package \texttt{stats}, on the observed-data log-likelihood.
% for the EM algorithm, and on the complete-data log-likelihood for the CEM algorithm.
In particular, \verb|optim| is initialized in the solution provided by the EM algorithm.
The optimization method of \citet{Byrd:Lu:Noce:Zhu:Alim:1995} is considered among the possible options of the \verb|optim| command.
As underlined by \citet{Loui:Find:Jour:1982} and \citet{Bold:Magn:Maxi:2009}, among others, the complete-data log-likelihood could be considered, instead of the observed-data log-likelihood, in order to simplify the computations because of its form as a sum of logarithms rather than a logarithm of a sum.
 
%\textbf{Naturally, in finding the covariance matrix of the estimated parameters for the CEM algorithm, the use of the complete-data log-likelihood, instead of the observed-data log-likelihood, simplifies the computations because of its form as a sum of logarithms rather than a logarithm of a sum \citep[see][]{Loui:Find:Jour:1982,Bold:Magn:Maxi:2009}. }

\section{Model selection and performance evaluation}
\label{sec:Model selection and performance evaluation}

The polynomial Gaussian CWM, in addition to $\boldsymbol{\psi}$, is also characterized by the polynomial degree ($r$) and by the number of components $k$.
So far, these quantities have been treated as \textit{a priori} fixed.
Nevertheless, for practical purposes, choosing a relevant model needs their choice.

\subsection{Bayesian information criterion and integrated completed likelihood}
\label{subsec:The Bayesian information criterion and the integrated completed likelihood}

A common way to select $r$ and $k$ consists in computing a convenient (likelihood-based) model selection criterion across a reasonable range of values for the couple $\left(r,k\right)$ and then choosing the couple associated to the best value of the adopted criterion.
Among the existing model selection criteria, the Bayesian information criterion \citep[BIC;][]{Schw:Esti:1978} and the integrated completed likelihood \citep[ICL;][]{Bier:Cele:Gova:Asse:2000} constitute the reference choices in the recent literature on mixture models.

The BIC is commonly used in model-based clustering and classifications applications involving a family of mixture models (\citealt{Fral:Raft:Mode:2002} and \citealt{McNi:Murp:Pars:2008}). 
The use of the BIC in mixture model selection was proposed by \citet{Dasg:Raft:Dete:1998}, based on an approximation to Bayes factors \citep{Kass:Raft:Baye:1995}.
In our context the BIC is given by
\begin{equation}
\text{BIC}=2l\left(\widehat{\boldsymbol{\psi}}\right)-\eta\ln n.
\label{eq:BIC}
\end{equation}
%where $\widehat{\boldsymbol{\psi}}$ is the ML estimate of $\boldsymbol{\psi}$.
\citet{Lero:Cons:1992} and \citet{Keri:Cons:2000} present theoretical results that, under certain regulatory conditions, support the use of the BIC for the estimation of the number of components in a mixture model. 

One potential problem with using the BIC for model selection in model-based classification or clustering applications is
that a mixture component does not necessarily correspond to a true cluster.
For example, a cluster might be represented by two mixture components.
In an attempt to focus model selection on clusters rather than mixture components, \citet{Bier:Cele:Gova:Asse:2000} introduced the ICL. 
The ICL, or the approximate ICL to be precise, is just the BIC penalized for estimated mean entropy and, in the classification framework used herein, it is given by
\begin{equation}
\text{ICL}\approx  \text{BIC} + \sum_{i=m+1}^n\sum_{j=1}^k \text{MAP}\left(\widehat{z}_{ij}\right)\ln \widehat{z}_{ij}
\label{eq:ICL}
\end{equation}
where $\sum_{i=m+1}^n\sum_{j=1}^k \text{MAP}\left(\widehat{z}_{ij}\right)\ln \widehat{z}_{ij}$ is the estimated mean entropy which reflects the uncertainty in the classification of observation $i$ into component $j$.
%, and 
%\begin{equation}
%\text{MAP}\left(\widehat{z}_{ij}\right)=\left\{
%\begin{array}{ccl}
%1 && \text{if $\max_j\left\{\widehat{z}_{ij}\right\}$ occurs at component $j$}\\
%0 && \text{otherwise}\\
%\end{array}
%\right.
%\label{eq:MAP}
%\end{equation}
%is the maximum \textit{a posteriori} classification given $\widehat{z}_{ij}$.
Therefore, the ICL should be less likely, compared to the BIC, to split one cluster into two mixture components, for example.

\citet[][p.~724]{Bier:Cele:Gova:Asse:2000}, based on numerical experiments, suggest to adopt the BIC and the ICL for indirect and direct applications, respectively. 

%Using the considerations -- based on numerical experiments -- made by \citet[][p.~724]{Bier:Cele:Gova:Asse:2000}, we can conclude this section by suggesting to adopt the BIC, when the concern of the polynomial Gaussian CWM is in terms of density estimation (indirect applications), and the ICL when the interest is in clustering (direct applications).   

\subsection{Adjusted Rand index}
\label{subsec:Adjusted Rand index}

Although the data analyses of Section~\ref{sec:illustrative examples} are mainly conducted as clustering examples, the true classifications are actually known for these data. 
In these examples, the Adjusted Rand Index \citep[ARI;][]{Hube:Arab:Comp:1985} is used to measure class agreement. 
The original Rand Index \citep[RI;][]{Rand:Obje:1971} is based on pairwise comparisons and is obtained by dividing the
number of pair agreements (observations that should be in the same group and are, plus those that should not be in the
same group and are not) by the total number of pairs.
%can be expressed as
%\begin{displaymath}
%	\frac{\text{number of agreements}}{\text{number of agreements + number of disagreements}},
%\end{displaymath}
%where the number of agreements and the number of disagreements are based on pairwise comparisons. 
RI assumes values on  $\left[0,1\right]$, where 0 indicates no pairwise agreements between the MAP classification and true group membership and 1 indicates perfect agreement.
One criticism of RI is that its expected value is greater than 0, making smaller values difficult to interpret.
ARI corrects RI for chance by allowing for the possibility that classification performed randomly should correctly classify some observations. 
Thus, ARI has an expected value of 0 and perfect classification would result in a value of 1.

\section{Illustrative examples and considerations}
\label{sec:illustrative examples}

This section begins showing some relations between polynomial Gaussian CWM and related models.
% finite mixtures of polynomial Gaussian regressions which, in this context, represent the natural competitors.
The section continues by looking at two applications on artificial and real data.

Parameters estimation for finite mixtures of polynomial Gaussian regressions is carried out via the \texttt{flexmix} function of the \texttt{R}-package \texttt{flexmix} which, among other, allows to perform EM and CEM algorithms.
If not otherwise stated, the number of repetitions, for the random initialization, will be fixed at $t=10$ for all the considered models.
%; the values $c_1=c_2=0.01$ will be adopted in \eqref{eq:variances constraints} with reference to the polynomial Gaussian CWM.

\subsection{Preliminary notes}
\label{subsec:preliminary considerations}

Before to illustrate the applications of the polynomial Gaussian CWM on real and artificial data sets, it is useful to show some limit cases in a direct application of this model.

Let 
\begin{equation}
p\left(y\left|x;\boldsymbol{\pi},\boldsymbol{\kappa}\right.\right)=\sum_{j=1}^k\pi_j\phi\left(y\left|x;\mu_r\left(x;\boldsymbol{\beta}_j\right),\sigma_{\varepsilon_j}^2\right.\right)
\label{eq:polynomial FMR of Y|X}
\end{equation}
be the (conditional) density of a finite mixture of polynomial Gaussian regressions of $Y$ on $x$.
Moreover, let 
\begin{equation}
p\left(x;\boldsymbol{\pi},\boldsymbol{\xi}\right)=\sum_{j=1}^k\pi_j\phi\left(x;\mu_{X|j},\sigma^2_{X|j}\right)
\label{eq:mixture of X}
\end{equation}
be the (marginal) density of a finite mixture of Gaussian distributions for $X$.
%By starting from the results of , in 
The following propositions gives two sufficient conditions under which the posterior probabilities of component-membership, arising from models \eqref{eq:polynomial FMR of Y|X} and \eqref{eq:mixture of X}, respectively, coincide with those from the polynomial Gaussian CWM.
Similar results for the linear Gaussian and the linear $t$ CWM are given in \citet{Ingr:Mino:Vitt:Loca:2011} and \citet{Ingr:Mino:Punz:Mode:2012}, respectively.
% in \eqref{eq:polynomial CWM} under convenient constraints on the model parameters.
%With the following propositions we will see as the posterior probabilities of component-membership, arising from models \eqref{eq:polynomial FMR of Y|X} and \eqref{eq:mixture of X}, coincide with those from the polynomial Gaussian CWM in \eqref{eq:polynomial CWM} under convenient constraints on the model parameters.
\begin{lem}
Given $k$, $r$, $\boldsymbol{\pi}$ and $\boldsymbol{\kappa}$, if
$\mu_{X|1}=\cdots=\mu_{X|k}=\mu_{X}$ and $\sigma_{X|1}=\cdots=\sigma_{X|k}=\sigma_{X}$, then model \eqref{eq:polynomial CWM} generates the same posterior probabilities of component-membership of model \eqref{eq:polynomial FMR of Y|X}.
\label{lem:polynomial CWM versus polynomial FMR}
\end{lem}
\begin{proof}
Given $k$, $r$, $\boldsymbol{\pi}$ and $\boldsymbol{\kappa}$, if the component marginal densities of $X$ do not depend from $j$, that is if $\mu_{X|1}=\cdots=\mu_{X|k}=\mu_{X}$ and $\sigma_{X|1}=\cdots=\sigma_{X|k}=\sigma_{X}$, then the posterior probabilities of component-membership for the CWM in \eqref{eq:polynomial CWM} can be written as 
\begin{eqnarray*}
P\left(\left.Z_{ij}=1\right|k,r,\boldsymbol{\pi},\boldsymbol{\kappa},\mu_{X},\sigma_{X}\right)&=&\frac{\pi_j\phi\left(y_i\left|x_i;\mu_r\left(x_i;\boldsymbol{\beta}_j\right),\sigma_{\varepsilon_j}^2\right.\right)\phi\left(x_i;\mu_{X},\sigma_{X}^2\right)}{\displaystyle\sum_{j=1}^k\pi_j\phi\left(y_i\left|x_i;\mu_r\left(x_i;\boldsymbol{\beta}_j\right),\sigma_{\varepsilon_j}^2\right.\right)\phi\left(x_i;\mu_{X},\sigma_{X}^2\right)}\\
&=&\frac{\pi_j\phi\left(y_i\left|x_i;\mu_r\left(x_i;\boldsymbol{\beta}_j\right),\sigma_{\varepsilon_j}^2\right.\right)\cancel{\phi\left(x_i;\mu_{X},\sigma_{X}^2\right)}}{\cancel{\phi\left(x_i;\mu_{X},\sigma_{X}^2\right)}\displaystyle\sum_{j=1}^k\pi_j\phi\left(y_i\left|x_i;\mu_r\left(x_i;\boldsymbol{\beta}_j\right),\sigma_{\varepsilon_j}^2\right.\right)}\\
&=&\frac{\pi_j\phi\left(y_i\left|x_i;\mu_r\left(x_i;\boldsymbol{\beta}_j\right),\sigma_{\varepsilon_j}^2\right.\right)}{\displaystyle\sum_{j=1}^k\pi_j\phi\left(y_i\left|x_i;\mu_r\left(x_i;\boldsymbol{\beta}_j\right),\sigma_{\varepsilon_j}^2\right.\right)},\quad \text{$i=1,\ldots,n$, $j=1,\ldots,k$,}
\end{eqnarray*}
which coincide with the posterior probabilities associated to the model in \eqref{eq:polynomial FMR of Y|X}.
\end{proof} 
\begin{lem}
Given $k$, $r$, $\boldsymbol{\pi}$ and $\boldsymbol{\xi}$, if
$\boldsymbol{\beta}_1=\cdots=\boldsymbol{\beta}_k=\boldsymbol{\beta}$ and $\sigma_{\varepsilon_1}=\cdots=\sigma_{\varepsilon_k}=\sigma_{\varepsilon}$, then model \eqref{eq:polynomial CWM} generates the same posterior probabilities of component-membership of model \eqref{eq:mixture of X}.
\label{lem:polynomial CWM versus mixtures of distributions}
\end{lem}
\begin{proof}
Given $k$, $r$, $\boldsymbol{\pi}$ and $\boldsymbol{\xi}$, if the component regression models does not depend from $j$, that is if $\boldsymbol{\beta}_1=\cdots=\boldsymbol{\beta}_k=\boldsymbol{\beta}$ and $\sigma_{\varepsilon_1}=\cdots=\sigma_{\varepsilon_k}=\sigma_{\varepsilon}$, then the posterior probabilities of component-membership for model \eqref{eq:polynomial CWM} can be written as 
\begin{eqnarray*}
P\left(\left.Z_{ij}=1\right|k,r,\boldsymbol{\pi},\boldsymbol{\xi},\boldsymbol{\beta},\sigma_{\varepsilon}\right)&=&\frac{\pi_j\phi\left(y_i\left|x_i;\mu_r\left(x_i;\boldsymbol{\beta}\right),\sigma_{\varepsilon}^2\right.\right)\phi\left(x_i;\mu_{X|j},\sigma_{X|j}^2\right)}{\displaystyle\sum_{j=1}^k\pi_j\phi\left(y_i\left|x_i;\mu_r\left(x_i;\boldsymbol{\beta}\right),\sigma_{\varepsilon}^2\right.\right)\phi\left(x_i;\mu_{X|j},\sigma_{X|j}^2\right)}\\
&=&\frac{\pi_j\cancel{\phi\left(y_i\left|x_i;\mu_r\left(x_i;\boldsymbol{\beta}\right),\sigma_{\varepsilon}^2\right.\right)}\phi\left(x_i;\mu_{X|j},\sigma_{X|j}^2\right)}{\cancel{\phi\left(y_i\left|x_i;\mu_r\left(x_i;\boldsymbol{\beta}\right),\sigma_{\varepsilon}^2\right.\right)}\displaystyle\sum_{j=1}^k\pi_j\phi\left(x_i;\mu_{X|j},\sigma_{X|j}^2\right)}\\
&=&\frac{\pi_j\phi\left(x_i;\mu_{X|j},\sigma_{X|j}^2\right)}{\displaystyle\sum_{j=1}^k\pi_j\phi\left(x_i;\mu_{X|j},\sigma_{X|j}^2\right)},\quad \text{$i=1,\ldots,n$, $j=1,\ldots,k$,}
\end{eqnarray*}
which coincide with the posterior probabilities of the model in \eqref{eq:mixture of X}.
\end{proof}

\subsection{Artificial data}
\label{subsec:Artificial data}

An artificial data set is here generated by a polynomial Gaussian CWM.
One of the aims is to highlight a situation in which a finite mixture of polynomial Gaussian regressions provides a wrong classification although the underlying groups are both well-separated and characterized by a polynomial Gaussian relationship of $Y$ on $x$.
%of underlying well-separated groups which a finite mixture of polynomial regressions is not able to find although the true model in each group is considered.
% the underlying group-structure in the data although it is very clear.

The data consist of $n=700$ bivariate observations randomly generated from a cubic ($r=3$) Gaussian CWM with $k=2$ groups having sizes $n_1=400$ and $n_2=300$.
\tablename~\ref{tab:artificial real parameters} reports the parameters of the generating model.
\begin{table}[!ht]
\caption{
Parameters of the cubic Gaussian CWM ($k=2$) used to generate the data.
\label{tab:artificial real parameters}
}
\centering
\subtable[Regression parameters]{
\label{tab:artificial real beta}
%\resizebox{!}{0.055\textheight}{
\begin{tabular}{crr}%{m{9mm}m{7mm}m{7mm}m{7mm}m{7mm}}
\toprule
	                       &	Component $j=1$	&	Component $j=2$	 \\
\midrule															
$\beta_{0j}$              &     0.000	      &	-8.000	       \\
$\beta_{1j}$              &    -1.000	      &	 0.100	       \\
$\beta_{2j}$              &     0.000	      &	-0.100	       \\
$\beta_{3j}$              &     0.100	      &	 0.150	       \\
$\sigma_{\varepsilon_j}$  &	  1.600    	   &	 2.300	       \\
\bottomrule	
\end{tabular}
%}
}
\qquad
\subtable[Other parameters]{
\label{tab:artificial real other}
%\resizebox{!}{0.055\textheight}{
\begin{tabular}{crr}%{m{9mm}m{7mm}m{7mm}m{7mm}m{7mm}}
\toprule
	                       &	Component $j=1$	&	Component $j=2$	 \\
\midrule															
$\pi_j$	                 &    0.571	      &	0.429	       \\
$\mu_{X|j}$               &	-2.000     	&	3.800	       \\
$\sigma_{X|j}$            &	 1.000     	&	0.700	       \\
\bottomrule	
\end{tabular}
%}
}
\end{table}
%Note that groups have the same regression coefficients. 
Simulated data are displayed in \figurename~\ref{fig:artificial-real-CW-plot} by what, from now on, will be simply named as CW-plot.   
\begin{figure}[!ht]
\centering
\resizebox{0.75\textwidth}{!}{
\includegraphics{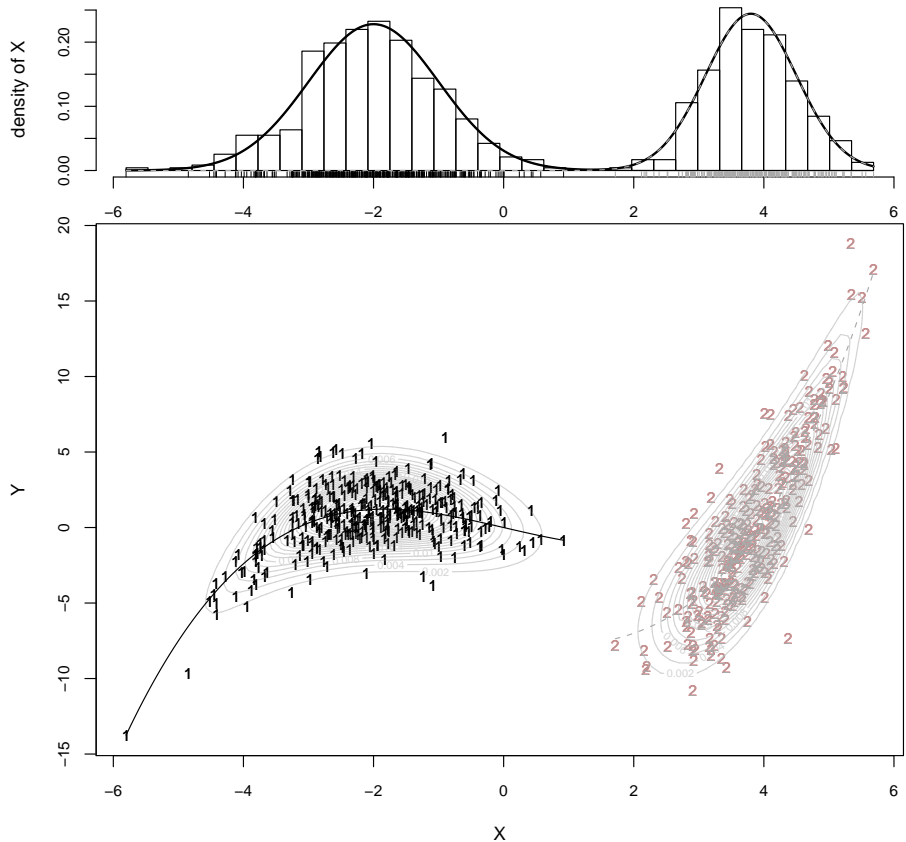} %[width=0.99\textwidth]
}
\caption{
CW-plot of the simulated data from a cubic Gaussian CWM ($n=700$ and $k=2$).
}
\label{fig:artificial-real-CW-plot}
\end{figure}
It allows to visualize the joint distribution of $\left(X,Y\right)'$, the regression of $Y$ on $x$, and the marginal distribution of $X$; as stressed from the beginning, these represent the key elements on which a CWM is based.  
A gray scale, and different line types (solid and dashed in this case), allow to distinguish the underlying groups in \figurename~\ref{fig:artificial-real-CW-plot}.
The top of \figurename~\ref{fig:artificial-real-CW-plot} is dedicated to the marginal distribution of $X$.
Here, we have an histogram of the simulated data on which are superimposed the component univariate Gaussian densities, multiplied by the corresponding weights $\pi_1$ and $\pi_2$, and the resulting mixture.
The scatter plot of the data is displayed at the bottom of \figurename~\ref{fig:artificial-real-CW-plot}.
The observations of the two groups are here differentiated by the labels $\mathsf{1}$ and $\mathsf{2}$ and by the different gray scales.
The true cubic Gaussian regressions are also separately represented.
%, in line with the top plot of \figurename~\ref{fig:artificial-real-CW-plot}, with different gray scales and different line types.
The underlying (generating) joint density is also visualized via isodensities; its 3D representation is displayed in \figurename~\ref{fig:artificial-real-density}. 
\begin{figure}[!ht]
\centering
\resizebox{0.8\textwidth}{!}{
\includegraphics{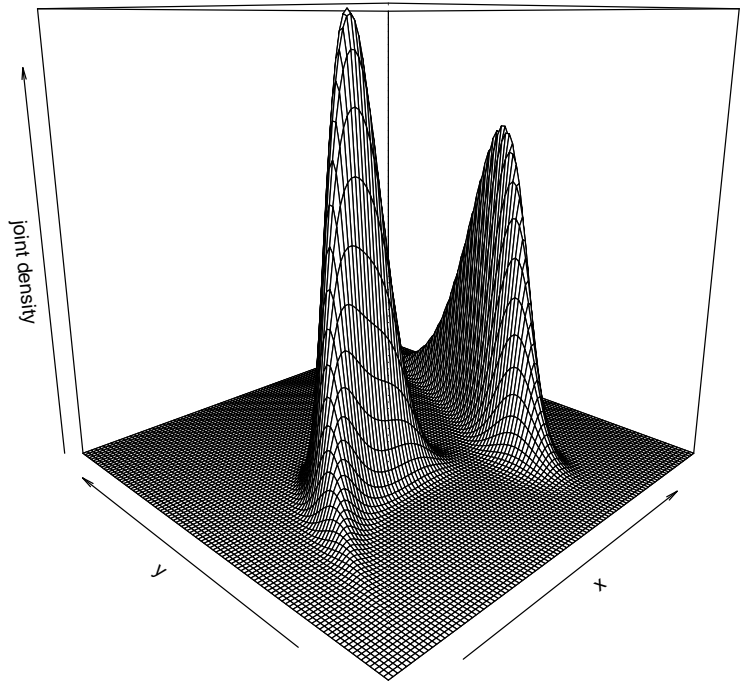} %[width=0.99\textwidth]
}
\caption{
Underlying joint density of the simulated data.
}
\label{fig:artificial-real-density}
\end{figure}

Now, we will suppose to forget the true classifications $\widetilde{\boldsymbol{z}}_i$, $i=1,\ldots,n$, and we will evaluate the performance of a finite mixture of cubic Gaussian regressions on these data.
To make \texttt{flexmix} in the best conditions, we have used the true classification as starting point in the EM algorithm and we have considered the true values of $k$ and $r$.
Nevertheless, without going into details about the estimated parameters, the clustering results are very bad, as confirmed by the scatter plot in \figurename~\ref{fig:artificial-FMRscatter} and by a very low value of the adjusted Rand Index ($\text{ARI}=0.088$).
\begin{figure}[!ht]
\centering
\resizebox{0.68\textwidth}{!}{
\includegraphics{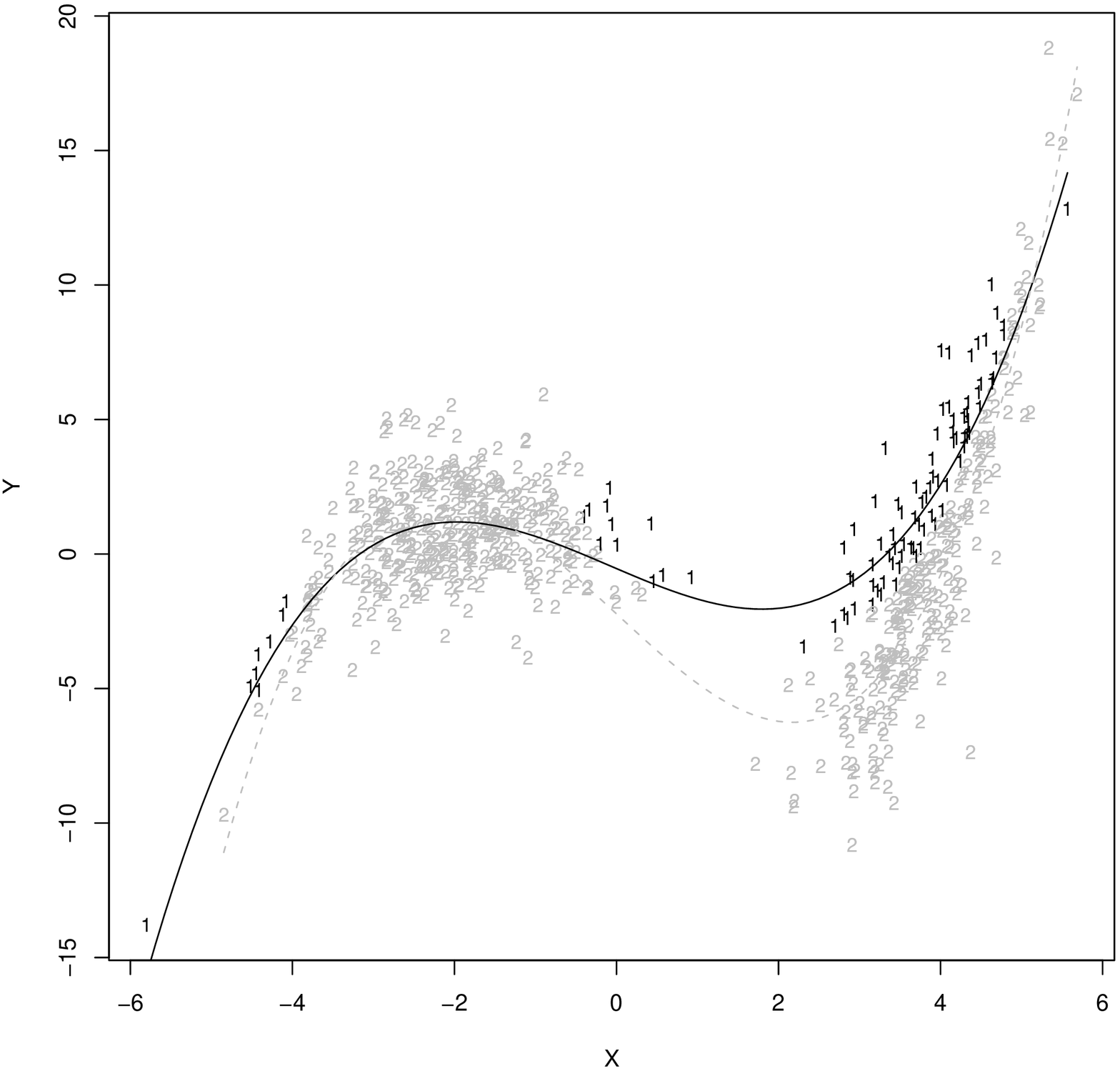} %[width=0.99\textwidth]
}
\caption{
Scatter plot of the artificial data with labels, and curves, arising from the ML-estimation (with the EM algorithm) of a finite mixture of $k=2$ cubic Gaussian regressions.
Plotting symbol and color for each observation is determined by the component with the maximum \textit{a posteriori} probability.
}
\label{fig:artificial-FMRscatter}
\end{figure}

On the contrary, as we shall see in a short time, the results obtained with the polynomial Gaussian CWM are optimal.
In particular, differently from the previous model, in this case we preliminarily estimate the number of mixture components $k$ and the polynomial degree $r$ according to what suggested in Section~\ref{subsec:The Bayesian information criterion and the integrated completed likelihood}.
\tablename~\ref{tab:likelihood-based information criteria} reports the values of $2l\left(\widehat{\boldsymbol{\psi}}\right)$, BIC, and ICL, obtained by using the EM algorithm, for a large enough number of couples $\left(k,r\right)$, $k=1,\ldots,5$ and $r=1,\ldots,5$. 
\begin{table}[!ht]
\caption{
Values of $2l\left(\widehat{\boldsymbol{\psi}}\right)$, BIC and ICL for a polynomial Gaussian CWM.
Different number of mixture components ($k=1,\ldots,5$) and of polynomial degrees ($r=1,\ldots,5$) are considered.
Bold numbers highlight the best values for BIC and ICL.
\label{tab:likelihood-based information criteria}
}
\centering
\subtable[$2l\left(\widehat{\boldsymbol{\psi}}\right)$]{\label{tab:loglik}
%\resizebox{!}{0.055\textheight}{
\begin{tabular}{cccccc}%{m{9mm}m{7mm}m{7mm}m{7mm}m{7mm}}
\toprule
\backslashbox{$k$}{$r$}     	&	1	&	2	& 3  &	4	&	5 \\
\midrule															
1 & -7356.441 & -7317.690 & -6582.602 & -6552.363 & -6546.713 \\ 
2 & -5855.160 & -5653.833 & -5636.396 & -5636.224 & -5633.650 \\ 
3 & -5713.095 & -5634.769 & -5629.208 & -5584.800 & -5584.004 \\ 
4 & -5672.787 & -5618.283 & -5599.250 & -5549.372 & -5549.227 \\ 
5 & -5661.108 & -5616.798 & -5579.644 & -5548.913 & -5546.247 \\ 
\bottomrule	
\end{tabular}
%}
}
\subtable[BIC]{
\label{tab:BIC}
%\resizebox{!}{0.055\textheight}{
\begin{tabular}{cccccc}%{m{9mm}m{7mm}m{7mm}m{7mm}m{7mm}}
\toprule
\backslashbox{$k$}{$r$}     	&	1	&	2	& 3  &	4	&	5 \\
\midrule															
1 & -7382.645 & -7350.446 & -6621.908 & -6598.220 & -6599.122 \\ 
2 & -5914.120 & -5725.895 & \textbf{-5721.560} & -5734.490 & -5745.019 \\ 
3 & -5804.810 & -5746.137 & -5760.229 & -5735.475 & -5754.332 \\ 
4 & -5797.258 & -5768.958 & -5776.129 & -5752.455 & -5778.515 \\ 
5 & -5818.334 & -5806.779 & -5802.380 & -5804.405 & -5834.495 \\ 
\bottomrule	
\end{tabular}
%}
}
\subtable[ICL]{
\label{tab:ICL}
%\resizebox{!}{0.055\textheight}{
\begin{tabular}{cccccc}%{m{9mm}m{7mm}m{7mm}m{7mm}m{7mm}}
\toprule
\backslashbox{$k$}{$r$}     	&	1	&	2	& 3  &	4	&	5 \\
\midrule															
1 & -7382.645 & -7350.446 & -6621.908 & -6598.220 & -6599.122 \\ 
2 & -5914.120 & -5726.044 & \textbf{-5721.563} & -5734.492 & -5745.198 \\ 
3 & -5847.224 & -5747.712 & -5873.706 & -5740.202 & -5754.593 \\ 
4 & -5906.636 & -5914.615 & -5783.383 & -5759.843 & -5790.514 \\ 
5 & -5923.268 & -5967.415 & -5851.145 & -5965.444 & -5876.014 \\ 
\bottomrule	
\end{tabular}
%}
}
\end{table}
Bold numbers highlight the best model according to each model selection criterion.
It is interesting to note as BIC and ICL select the same (true) model characterized by $k=2$ and $r=3$. 
For this model, \tablename~\ref{tab:artificial-EM-parameters} shows the ML estimated parameters (and their standard errors in round brackets)
%, computed according to what said in Section~\ref{subsec:Estimating the degrees of freedom}) 
while \figurename~\ref{fig:artificial-CWMscatter} displays the resulting scatter plot. 
Visibly, optimal results in terms of fit and clustering are obtained, as also corroborated by the value $\text{ARI}=1$.  
\begin{table}[!ht]
\caption{
Maximum likelihood estimated parameters, obtained with the EM algorithm, for a cubic Gaussian CWM ($k=2$).
Standard Errors are displayed in round brackets.
\label{tab:artificial-EM-parameters}
}
\centering
\subtable[Regression parameters]{\label{tab:artificial-EM-beta}
%\resizebox{!}{0.055\textheight}{
\begin{tabular}{crr}%{m{9mm}m{7mm}m{7mm}m{7mm}m{7mm}}
\toprule
	                       &	Component $j=1$	&	Component $j=2$	 \\
\midrule															
$\widehat{\beta}_{0j}$              &    \begin{tabular}{c} 0.030\\ (0.270) \end{tabular}	      &	     	\begin{tabular}{c} -8.042\\ (10.391) \end{tabular}	    \\[4mm]
$\widehat{\beta}_{1j}$              &    \begin{tabular}{c} -1.004\\ (0.398) \end{tabular}	      &	     	\begin{tabular}{c} -0.090\\ (8.558) \end{tabular}	    \\[4mm]
$\widehat{\beta}_{2j}$              &    \begin{tabular}{c} 0.041\\ (0.194) \end{tabular}	      &	     	\begin{tabular}{c} -0.038\\ (2.304) \end{tabular}	    \\[4mm]
$\widehat{\beta}_{3j}$              &    \begin{tabular}{c} 0.114\\ (0.028) \end{tabular}	      &	     	\begin{tabular}{c} 0.147\\ (0.202) \end{tabular}	    \\[4mm]
$\widehat{\sigma}_{\varepsilon_j}$  &    \begin{tabular}{c} 1.586\\ (0.056) \end{tabular}	      &	     	\begin{tabular}{c} 2.502\\ (0.102) \end{tabular}	    \\
\bottomrule	
\end{tabular}
%}
}
%\hspace{-4mm}
\subtable[Other parameters]{\label{tab:artificial-EM-other}
%\resizebox{!}{0.055\textheight}{
\begin{tabular}{crr}%{m{9mm}m{7mm}m{7mm}m{7mm}m{7mm}}
\toprule
	                       &	Component $j=1$	&	Component $j=2$	 \\
\midrule															
$\widehat{\pi}_j$	                 &    \begin{tabular}{c} 0.571\\ (0.028) \end{tabular}	      &	     	\begin{tabular}{c} 0.429\\ (0.024) \end{tabular}	    \\[4mm]
$\widehat{\mu}_{X|j}$               &    \begin{tabular}{c} -2.046\\ (0.051) \end{tabular}	      &	     	\begin{tabular}{c} 3.814\\ (0.039) \end{tabular}	    \\[4mm]
$\widehat{\sigma}_{X|j}$            &    \begin{tabular}{c} 1.022\\ (0.036) \end{tabular}	      &	     	\begin{tabular}{c} 0.682\\ (0.027) \end{tabular}	    \\
\bottomrule	
\end{tabular}
%}
}
\end{table}
\begin{figure}[!ht]
\centering
\resizebox{0.68\textwidth}{!}{
\includegraphics{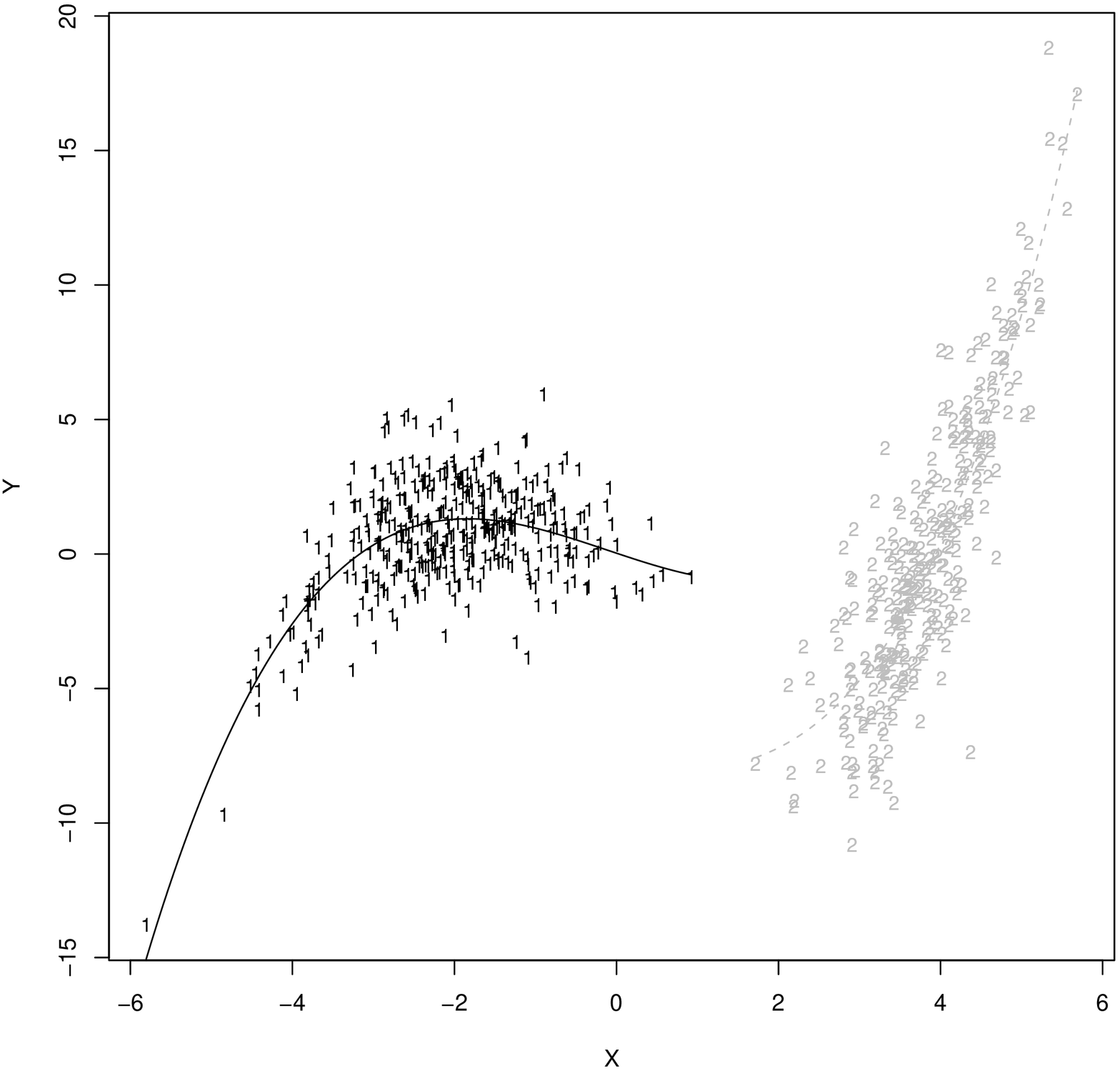} %[width=0.99\textwidth]
}
\caption{
Scatter plot of the artificial data with labels, and curves, arising from the ML-estimation (with the EM algorithm) of a cubic Gaussian CWM with $k=2$.
}
\label{fig:artificial-CWMscatter}
\end{figure}
%From a simple comparison between \figurename~\ref{fig:artificial-real-CW-plot} and \figurename~\ref{fig:artificial-CWMscatter}, we can note the optimal results in terms of fit and clustering, as also corroborated by the value $\text{ARI}=1$.  

Other experiments, whose results are not reported here for brevity's sake, have shown that a finite mixture of polynomial Gaussian regressions is not able to find the underlying group-structure when it also affects the marginal distribution of $X$.
%The interesting fact is that this consideration remains true also when the 
%As seen before, this consideration  
These results generalize to the case $r>1$ the considerations that \cite{Ingr:Mino:Vitt:Loca:2011} make with reference to the linear Gaussian CWM in comparison with finite mixtures of linear Gaussian regressions.

\subsection{Real data}
\label{subsec:Real data}

The ``places'' data from the \textit{Places Rated Almanac} \citep{Boye:Sava:Plac:1985} are a collection of nine composite variables constructed for $n=329$ metropolitan areas of the United States in order to measure the quality of life.
There are $k=2$ groups of places, small (group 1) and large (group 2), with the first of size 303.
For the current purpose, we only use the two variables $X=$``health care and environment'' and $Y=$``arts and cultural facilities'' measured so that the higher the score, the better.
These are the two variables having the highest correlation \citep[][]{Kope:Hoff:Gene:1992}. 
\figurename~\ref{fig:places-real scatter} displays the scatter plot of $X$ versus $Y$ in both groups.
\begin{figure}[!ht]
\centering
\resizebox{0.68\textwidth}{!}{
\includegraphics{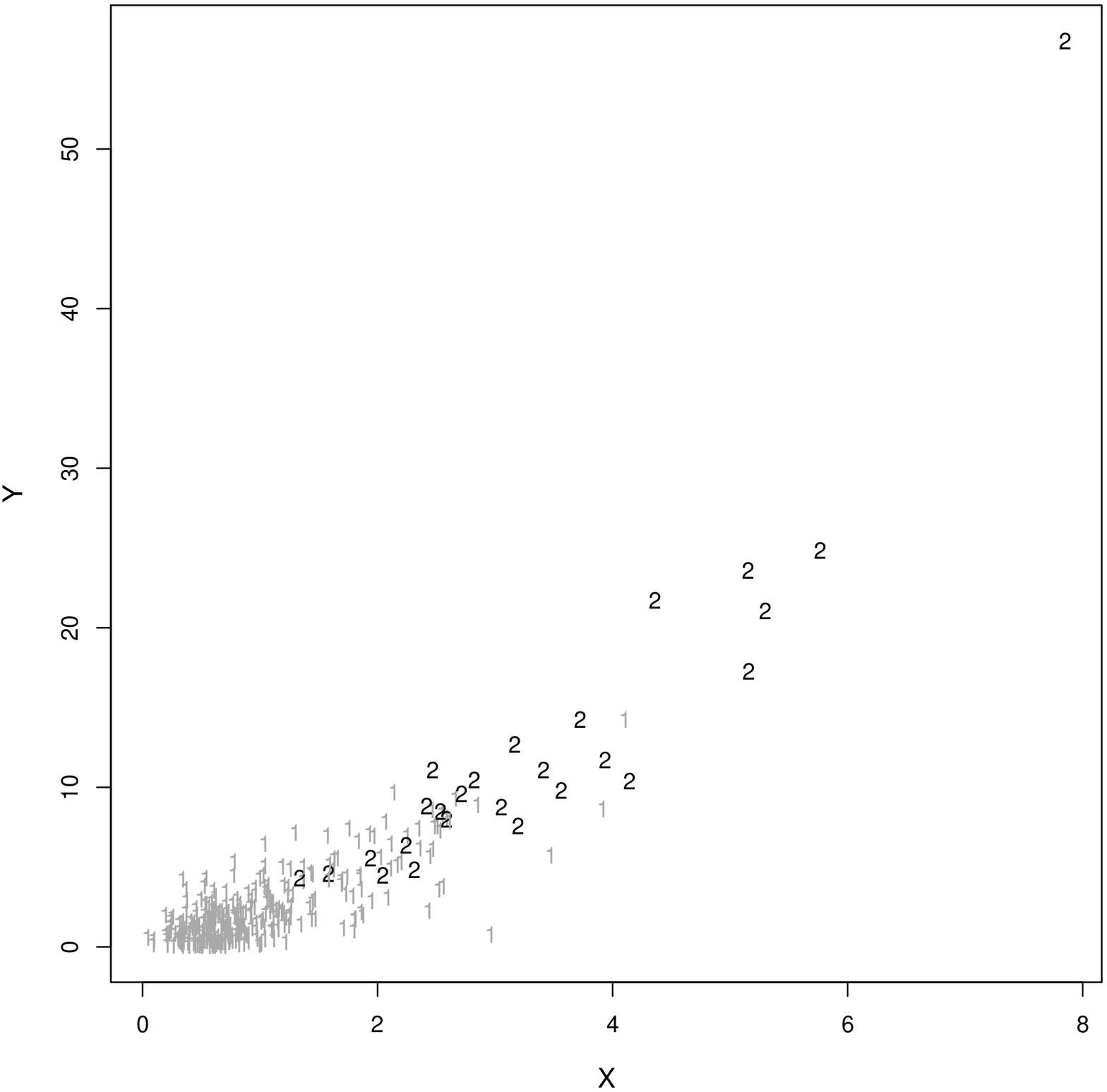} %[width=0.99\textwidth]
}
\caption{
Scatter plot of $X=$``health care and environment'' and $Y=$``arts and cultural facilities'' for 303 small (denoted with $\mathsf{1}$) and 26 large (denoted with $\mathsf{2}$) metropolitan areas of the United States.
Data from \citet{Boye:Sava:Plac:1985}.
}
\label{fig:places-real scatter}
\end{figure}
In view of making clustering, the situation seems more complicated than the previous one due to a prominent overlapping between groups; this is an aspect that have to be taken into account in evaluating the quality of the clustering results. 
Furthermore, it is possible to see a clear parabolic functional relationship of $Y$ on $x$ in group 2.
Regarding group 1, \tablename~\ref{tab:places-parabolic group 1} shows the summary results of a polynomial regression ($r=5$) as provided by the \texttt{lm} function of the \texttt{R}-package \texttt{stats}.
\begin{table}[!ht]
\caption{
Estimated parameters, and corresponding summary statistics, of a polynomial regression (of degree $r=5$) fitted on group 1 via the \texttt{R}-function \texttt{lm}. 
\label{tab:places-parabolic group 1}
}
\centering
%\resizebox{!}{0.055\textheight}{
\begin{tabular}{crrrr}
\toprule
	             &	Estimate	  & Std. error  & $t$-value & $p$-value	 \\
\midrule															
$\beta_{01}$    &  0.256     & 0.705       & 0.363     & 0.71687      \\
$\beta_{11}$    &  4.024     & 3.215       & 1.251     & 0.21176      \\
$\beta_{21}$    &  -7.460    & 5.033       & -1.482    & 0.13937      \\
$\beta_{31}$    &  7.580     & 3.373       & 2.247     & 0.02537      \\
$\beta_{41}$    &  -2.762    & 0.993       & -2.782    & 0.00574      \\
$\beta_{51}$    &  0.328     & 0.104       & 3.141     & 0.00185      \\
\bottomrule	
\end{tabular}
%}
\end{table}
At a (common) nominal level of 0.05, the only significant parameters appear to be those related to the parabolic function.
   
Now, we will suppose to forget the true classification in small and large places and we will try to estimate it by directly considering the case $k=2$.
\figurename~\ref{fig:Lik-based-functions} displays the values of $2l\left(\widehat{\boldsymbol{\psi}}\right)$, BIC, and ICL for the polynomial Gaussian CWM in correspondence of $r$ ranging from 1 to 8. 
\begin{figure}[!ht]
\centering
\resizebox{0.54\textwidth}{!}{
\includegraphics{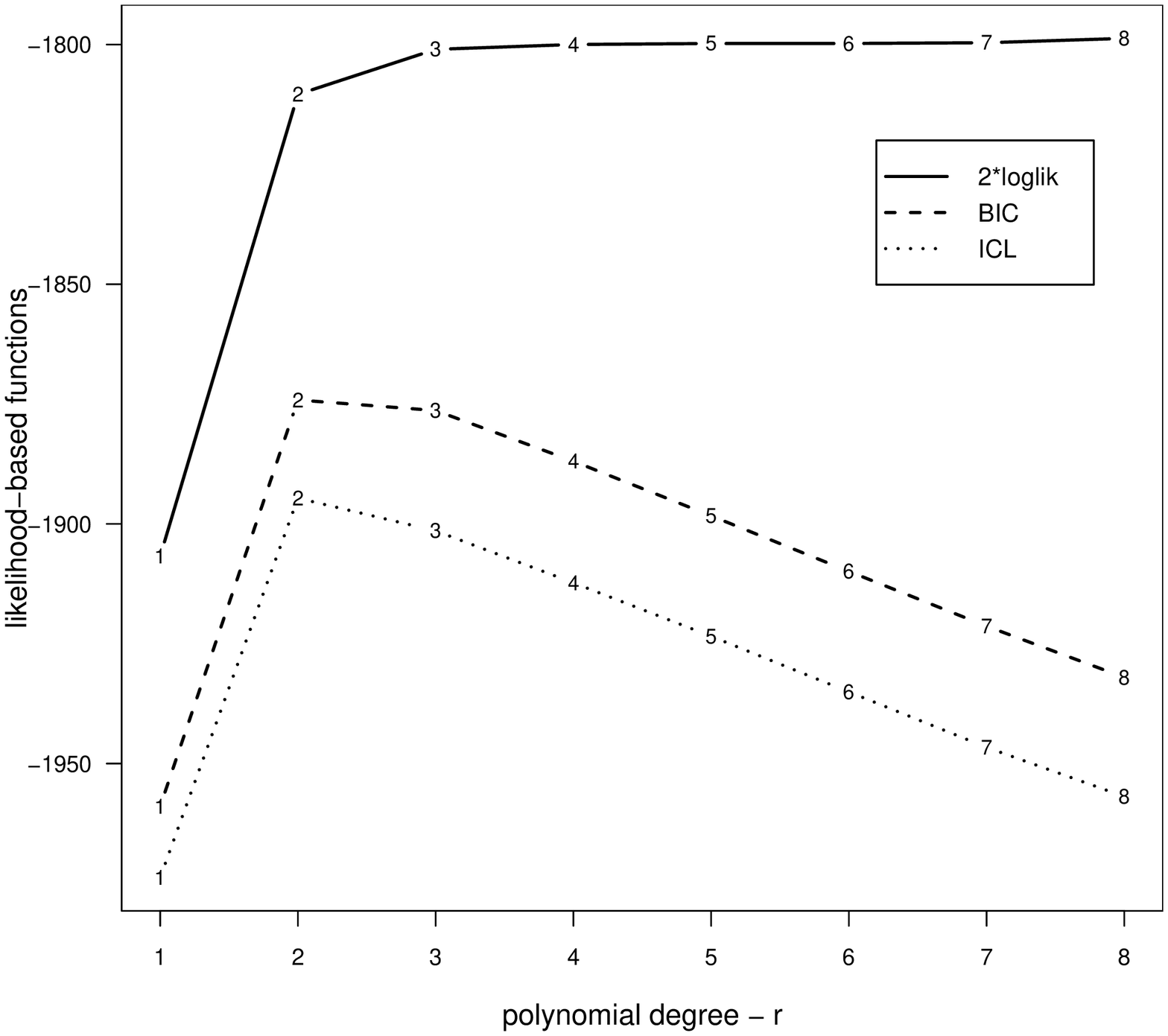} %[width=0.99\textwidth]
}
\caption{
Values of $2l\left(\widehat{\boldsymbol{\psi}}\right)$, BIC and ICL, in correspondence of $r=1,\ldots,8$, in the polynomial Gaussian CWM. 
}
\label{fig:Lik-based-functions}
\end{figure}
The best polynomial degree, according to the BIC (-1874.148) and the ICL (-1894.629), is $r=2$.
This result corroborates the above considerations.
\tablename~\ref{tab:places-EM-parameters} summarizes the parameter estimates, and the corresponding standard errors, for the quadratic Gaussian CWM with $k=2$.
\begin{table}[!ht]
\caption{
Parameters of a quadratic Gaussian CWM estimated with the EM algorithm ($k=2$).
Standard errors are displayed in round brackets.
\label{tab:places-EM-parameters}
}
\centering
\subtable[Regression parameters]{\label{tab:places-CEM-beta}
%\resizebox{!}{0.055\textheight}{
\begin{tabular}{crr}%{m{9mm}m{7mm}m{7mm}m{7mm}m{7mm}}
\toprule
	                       &	Component $j=1$	&	Component $j=2$	 \\
\midrule															
$\widehat{\beta}_{0j}$              &    \begin{tabular}{c} 0.677\\ (0.349) \end{tabular}	      &	     	\begin{tabular}{c} 4.503\\ (0.877) \end{tabular}	    \\[4mm]
$\widehat{\beta}_{1j}$              &    \begin{tabular}{c} 0.191\\ (1.081) \end{tabular}	      &	     	\begin{tabular}{c} -1.821\\ (0.592) \end{tabular}	    \\[4mm]
$\widehat{\beta}_{2j}$              &    \begin{tabular}{c} 0.892\\ (0.787) \end{tabular}	      &	     	\begin{tabular}{c} 1.019\\ (0.084) \end{tabular}	    \\[4mm]
$\widehat{\sigma}_{\varepsilon_j}$  &    \begin{tabular}{c} 0.866\\ (0.058) \end{tabular}	      &	     	\begin{tabular}{c} 2.220\\ (0.159) \end{tabular}	    \\
\bottomrule	
\end{tabular}
%}
}
\quad
\subtable[Other parameters]{\label{tab:places-CEM-other}
%\resizebox{!}{0.055\textheight}{
\begin{tabular}{crr}%{m{9mm}m{7mm}m{7mm}m{7mm}m{7mm}}
\toprule
	                       &	Component $j=1$	&	Component $j=2$	 \\
\midrule															
$\widehat{\pi}_j$	                 &    \begin{tabular}{c} 0.662\\ (0.050) \end{tabular}	      &	     	\begin{tabular}{c} 0.338\\ (0.039) \end{tabular}	    \\[4mm]
$\widehat{\mu}_{X|j}$               &    \begin{tabular}{c} 0.699\\ (0.024) \end{tabular}	      &	     	\begin{tabular}{c} 2.138\\ (0.136) \end{tabular}	    \\[4mm]
$\widehat{\sigma}_{X|j}$            &    \begin{tabular}{c} 0.294\\ (0.019) \end{tabular}	      &	     	\begin{tabular}{c} 1.194\\ (0.082) \end{tabular}	    \\
\bottomrule	
\end{tabular}
%}
}
\end{table}
The CW-plot is displayed in \figurename~\ref{fig:places-CWMEMscatter}. 
\begin{figure}[!ht]
\centering
\resizebox{0.75\textwidth}{!}{
\includegraphics{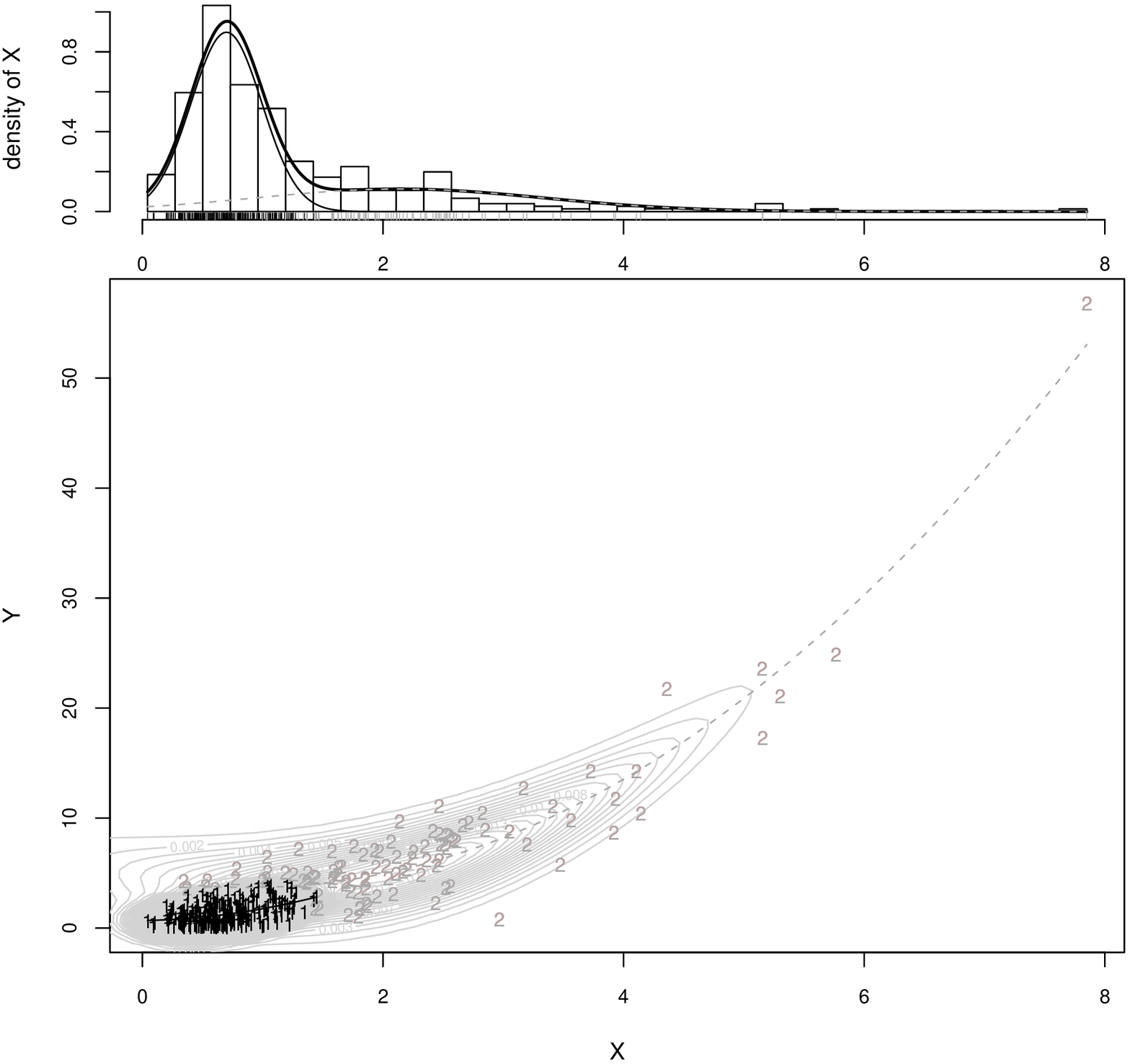} %[width=0.99\textwidth]
}
\caption{
CW-plot of the quadratic Gaussian CWM fitted, via the EM algorithm, on the ``places'' data ($k=2$).
Plotting symbol and color for each observation is determined by the component with the maximum \textit{a posteriori} probability.
}
\label{fig:places-CWMEMscatter}
\end{figure}
The ARI results to be 0.208; the corresponding value for a finite mixture of $k=2$ quadratic Gaussian regressions is 0.146.
Note that, the ARI increases up to 0.235 if the quadratic Gaussian CWM is fitted via the CEM algorithm.  
% 0.207768
% 0.1463936

\subsubsection{Classification evaluation}
\label{subsubsec:Classification evaluation}

Now, suppose to be interested in evaluating the impact of possible $m$ labeled data, with $m<n$, on the classification of the remaining $n-m$ unlabeled observations.
With this aim, using the same data, we have performed a simple simulation study with the following scheme.
For each $m$ ranging from 1 to 250 (250/329=0.760), we have randomly generated 500 vectors, of size $m$, with elements indicating the observations to consider as labeled (using the true labels for them).
The discrete uniform distribution, taking values on the set $\left\{1,\ldots,n\right\}$, was used as generating model.
In each of the 500 replications, the quadratic Gaussian CWM (with $k=2$) was fitted, with the EM algorithm, to classify the $n-m$ unlabeled observations.
\figurename~\ref{fig:places - ARI} shows the average ARI values, computed across the 500 replications only on the $n-m$ unlabeled observations, for each considered value of $m$.
\begin{figure}[!ht]
\centering
\resizebox{0.54\textwidth}{!}{
\includegraphics{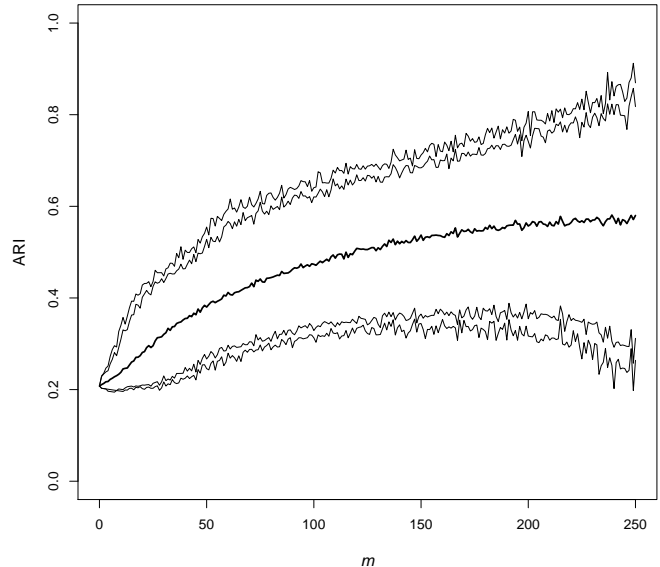} %[width=0.99\textwidth]
}
\caption{
Average ARI values, with respect to 500 replications, for each considered value of $m$ (central line with the highest width).
The simulated quantiles of probability 0.025, 0.05, 0.95 and 0.975, are also superimposed.   
}
\label{fig:places - ARI}
\end{figure}
To give an idea of the conditional variability of the obtained results, quantiles of probability 0.025, 0.05, 0.95 and 0.975, are also superimposed on the plot.
As expected, the knowledge of the true labels for $m$ observations tends (on average, as measured by the ARI) to improve the quality of the classification for the remaining observations when $m$ increases.
  
\section{Discussion and future work}
\label{sec:conclusions}

An extension of the linear Gaussian CWM has been introduced which allows to model possible nonlinear relationships in each mixture component through a polynomial function. 
This model, named polynomial Gaussian CWM, was justified on the ground of density estimation, model-based clustering, and model-based classification.
Parameter estimation was carried out within the EM algorithm framework, and the BIC and the ICL were used for model selection.
Theoretical arguments were also given showing as the posterior probabilities of component-membership arising from some well-known mixture models can be obtained by conveniently constrained the parameters of the polynomial Gaussian CWM.  
The proposed model was then applied on a real and an artificial data set. 
Here, excellent clustering/classification performance was achieved when compared with existing model-based clustering techniques.

Future work will follow several avenues.
To begin, extension of the polynomial Gaussian CWM, to more than two variables, will be considered.
This extension could initially concern only the $X$ variable and then involve also the $Y$ one.
%Here, the idea is to preliminarly take into account  the multivariate counterpart of the actual univariate $X$-variable and then trying to extend also to the $Y$-variable. 
Afterwards, the polynomial $t$ CWM could be introduced by simply substituting the Gaussian density with a Student-$t$ distribution providing more robust inference for data characterized by noise and outliers \citep{Lang:Litt:Tayl:Robu:1989}.
Finally, it could be interesting to evaluate, theoretically or by simulation, the most convenient model selection criteria for the proposed model.
Note that this search, in line with \citet[][p.~724]{Bier:Cele:Gova:Asse:2000}, could be separately conducted according to the type of application, direct or indirect, of the model. 

\section*{Acknowledgements}
The author sincerely thanks Salvatore Ingrassia for helpful comments and suggestions.    

\bibliographystyle{natbib}
\bibliography{References}

\begin{thebibliography}{}

\bibitem[Aitken(1926)]{Aitk:OnBe:1926}
Aitken, A. (1926).
\newblock {On Bernoulli's numerical solution of algebraic equations}.
\newblock In {\em Proceedings of the Royal Society of Edinburgh}, volume~46,
  pages 289--305.

\bibitem[Andrews {\em et~al.}(2011)]{Andr:McNi:Sube:Mode:2011}
Andrews, J., McNicholas, P., and Subedi, S. (2011).
\newblock Model-based classification via mixtures of multivariate
  $t$-distributions.
\newblock {\em Computational Statistics and Data Analysis}, {\bf 55}(1),
  520--529.

\bibitem[Banfield and Raftery(1993)]{Banf:Raft:mode:1993}
Banfield, J.~D. and Raftery, A.~E. (1993).
\newblock Model-based {G}aussian and non-{G}aussian clustering.
\newblock {\em Biometrics}, {\bf 49}(3), 803--821.

\bibitem[Biernacki {\em et~al.}(2000)]{Bier:Cele:Gova:Asse:2000}
Biernacki, C., Celeux, G., and Govaert, G. (2000).
\newblock Assessing a mixture model for clustering with the integrated
  completed likelihood.
\newblock {\em Pattern Analysis and Machine Intelligence, IEEE Transactions
  on}, {\bf 22}(7), 719--725.

\bibitem[Biernacki {\em et~al.}(2003)]{Bier:Cele:Gova:Choo:2003}
Biernacki, C., Celeux, G., and Govaert, G. (2003).
\newblock Choosing starting values for the {EM} algorithm for getting the
  highest likelihood in multivariate {G}aussian mixture models.
\newblock {\em Computational Statistics \& Data Analysis}, {\bf 41}(3-4),
  561--575.

\bibitem[B{\"o}hning {\em et~al.}(1994)]{Bohn:Diet:Scha:Schl:Lind:TheD:1994}
B{\"o}hning, D., Dietz, E., Schaub, R., Schlattmann, P., and Lindsay, B.
  (1994).
\newblock The distribution of the likelihood ratio for mixtures of densities
  from the one-parameter exponential family.
\newblock {\em Annals of the Institute of Statistical Mathematics}, {\bf
  46}(2), 373--388.

\bibitem[Boldea and Magnus(2009)]{Bold:Magn:Maxi:2009}
Boldea, O. and Magnus, J. (2009).
\newblock Maximum likelihood estimation of the multivariate normal mixture
  model.
\newblock {\em Journal of the American Statistical Association}, {\bf
  104}(488), 1539--1549.

\bibitem[Byrd {\em et~al.}(1995)]{Byrd:Lu:Noce:Zhu:Alim:1995}
Byrd, R., Lu, P., Nocedal, J., and Zhu, C. (1995).
\newblock A limited memory algorithm for bound constrained optimization.
\newblock {\em SIAM Journal on Scientific Computing}, {\bf 16}(5), 1190--1208.

\bibitem[Celeux and Govaert(1992)]{Cele:Gova:Acla:1992}
Celeux, G. and Govaert, G. (1992).
\newblock A classification {EM} algorithm for clustering and two stochastic
  versions.
\newblock {\em Computational Statistics \& Data Analysis}, {\bf 14}(3),
  315--332.

\bibitem[Celeux and Govaert(1995)]{Cele:Gova:Gaus:1995}
Celeux, G. and Govaert, G. (1995).
\newblock Gaussian parsimonious clustering models.
\newblock {\em Pattern Recognition}, {\bf 28}(5), 781--793.

\bibitem[Dasgupta and Raftery(1998)]{Dasg:Raft:Dete:1998}
Dasgupta, A. and Raftery, A. (1998).
\newblock Detecting features in spatial point processes with clutter via
  model-based clustering.
\newblock {\em Journal of the American Statistical Association}, {\bf 93}(441),
  294--302.

\bibitem[Dean {\em et~al.}(2006)]{Dean:Murp:Down:Usin:2006}
Dean, N., Murphy, T., and Downey, G. (2006).
\newblock Using unlabelled data to update classification rules with
  applications in food authenticity studies.
\newblock {\em Journal of the Royal Statistical Society: Series C (Applied
  Statistics)}, {\bf 55}(1), 1--14.

\bibitem[Dempster {\em et~al.}(1977)]{Demp:Lair:Rubi:Maxi:1977}
Dempster, A., Laird, N., and Rubin, D. (1977).
\newblock {Maximum likelihood from incomplete data via the EM algorithm}.
\newblock {\em Journal of the Royal Statistical Society. Series B
  (Methodological)}, {\bf 39}(1), 1--38.

\bibitem[Escobar and West(1995)]{Esco:West:Baye:1995}
Escobar, M. and West, M. (1995).
\newblock Bayesian density estimation and inference using mixtures.
\newblock {\em Journal of the American Statistical Association}, {\bf 90}(430),
  577--588.

\bibitem[Fraley and Raftery(2002)]{Fral:Raft:Mode:2002}
Fraley, C. and Raftery, A. (2002).
\newblock Model-based clustering, discriminant analysis, and density
  estimation.
\newblock {\em Journal of the American Statistical Association}, {\bf 97}(458),
  611--631.

\bibitem[Fraley and Raftery(1998)]{Fral:Raft:Howm:1998}
Fraley, C. and Raftery, A.~E. (1998).
\newblock How many clusters? which clustering method? answers via model-based
  cluster analysis.
\newblock {\em Computer Journal}, {\bf 41}(8), 578--588.

\bibitem[Fr{\"u}hwirth-Schnatter(2006)]{Fruh:Fine:2006}
Fr{\"u}hwirth-Schnatter, S. (2006).
\newblock {\em {Finite mixture and Markov switching models}}.
\newblock Springer, New York.

\bibitem[Gershenfeld(1997)]{Gers:Nonl:1997}
Gershenfeld, N. (1997).
\newblock Nonlinear inference and cluster-weighted modeling.
\newblock {\em Annals of the New York Academy of Sciences}, {\bf 808}(1),
  18--24.

\bibitem[Greselin and Ingrassia(2010)]{Gres:Ingr:Cons:2010}
Greselin, F. and Ingrassia, S. (2010).
\newblock Constrained monotone {EM} algorithms for mixtures of multivariate $t$
  distributions.
\newblock {\em Statistics and computing}, {\bf 20}(1), 9--22.

\bibitem[Gr{\"u}n and Leisch(2008)]{Grun:Leis:Flex:2008}
Gr{\"u}n, B. and Leisch, F. (2008).
\newblock Flexmix version 2: Finite mixtures with concomitant variables and
  varying and constant parameters.
\newblock {\em Journal of Statistical Software}, {\bf 28}(4), 1--35.

\bibitem[Gutierrez {\em et~al.}(1995)]{Guti:Carr:Wang:Lee:Tayl:Anal:1995}
Gutierrez, R., Carroll, R., Wang, N., Lee, G., and Taylor, B. (1995).
\newblock Analysis of tomato root initiation using a normal mixture
  distribution.
\newblock {\em Biometrics}, {\bf 51}(4), 1461--1468.

\bibitem[Hosmer~Jr.(1973)]{Hosm:acom:1973}
Hosmer~Jr., D. (1973).
\newblock A comparison of iterative maximum likelihood estimates of the
  parameters of a mixture of two normal distributions under three different
  types of sample.
\newblock {\em Biometrics}, {\bf 29}(4), 761--770.

\bibitem[Hubert and Arabie(1985)]{Hube:Arab:Comp:1985}
Hubert, L. and Arabie, P. (1985).
\newblock Comparing partitions.
\newblock {\em Journal of Classification}, {\bf 2}(1), 193--218.

\bibitem[Ingrassia {\em et~al.}(2012a)]{Ingr:Mino:Vitt:Loca:2011}
Ingrassia, S., Minotti, S.~C., and Vittadini, G. (2012a).
\newblock Local statistical modeling via the cluster-weighted approach with
  elliptical distributions.
\newblock {\em Journal of Classification}, {\bf 29}(3), ..--..

\bibitem[Ingrassia {\em et~al.}(2012b)]{Ingr:Mino:Punz:Mode:2012}
Ingrassia, S., Minotti, S.~C., and Punzo, A. (2012b).
\newblock Model-based clustering via linear cluster-weighted models.
\newblock arXiv.org e-print 1206.3974, http://arxiv.org/abs/1206.3974.

\bibitem[Karlis and Santourian(2009)]{Karl:Sant:Mode:2009}
Karlis, D. and Santourian, A. (2009).
\newblock Model-based clustering with non-elliptically contoured distributions.
\newblock {\em Statistics and Computing}, {\bf 19}(1), 73--83.

\bibitem[Karlis and Xekalaki(2003)]{Karl:Xeka:Choo:2003}
Karlis, D. and Xekalaki, E. (2003).
\newblock {Choosing initial values for the EM algorithm for finite mixtures}.
\newblock {\em Computational Statistics \& Data Analysis}, {\bf 41}(3--4),
  577--590.

\bibitem[Kass and Raftery(1995)]{Kass:Raft:Baye:1995}
Kass, R. and Raftery, A. (1995).
\newblock Bayes factors.
\newblock {\em Journal of the American Statistical Association}, {\bf 90}(430),
  773--795.

\bibitem[Keribin(2000)]{Keri:Cons:2000}
Keribin, C. (2000).
\newblock Consistent estimation of the order of mixture models.
\newblock {\em Sankhy{\=a}: The Indian Journal of Statistics, Series A}, {\bf
  62}(1), 49--66.

\bibitem[Kopalle and Hoffman(1992)]{Kope:Hoff:Gene:1992}
Kopalle, P. and Hoffman, D. (1992).
\newblock Generalizing the sensitivity conditions in an overall index of
  product quality.
\newblock {\em Journal of Consumer Research}, {\bf 18}(4), 530--535.

\bibitem[Lange {\em et~al.}(1989)]{Lang:Litt:Tayl:Robu:1989}
Lange, K.~L., Little, R. J.~A., and Taylor, J. M.~G. (1989).
\newblock Robust statistical modeling using the $t$ distribution.
\newblock {\em Journal of the American Statistical Association}, {\bf 84}(408),
  881--896.

\bibitem[Leisch(2004)]{Leis:Flex:2004}
Leisch, F. (2004).
\newblock Flexmix: A general framework for finite mixture models and latent
  class regression in {R}.
\newblock {\em Journal of Statistical Software}, {\bf 11}(8), 1--18.

\bibitem[Leroux(1992)]{Lero:Cons:1992}
Leroux, B. (1992).
\newblock Consistent estimation of a mixing distribution.
\newblock {\em The Annals of Statistics}, {\bf 20}(3), 1350--1360.

\bibitem[Lin(2009)]{Lin:Maxi:2009}
Lin, T. (2009).
\newblock Maximum likelihood estimation for multivariate skew normal mixture
  models.
\newblock {\em Journal of Multivariate Analysis}, {\bf 100}(2), 257--265.

\bibitem[Lin(2010)]{Lin:Robu:2010}
Lin, T. (2010).
\newblock Robust mixture modeling using multivariate skew $t$ distributions.
\newblock {\em Statistics and Computing}, {\bf 20}(3), 343--356.

\bibitem[Lo {\em et~al.}(2008)]{Lo:Brin:Gott:Cyto:2008}
Lo, K., Brinkman, R., and Gottardo, R. (2008).
\newblock Automated gating of flow cytometry data via robust model-based
  clustering.
\newblock {\em Cytometry Part A}, {\bf 73}(4), 321--332.

\bibitem[Louis(1982)]{Loui:Find:Jour:1982}
Louis, T. (1982).
\newblock Finding the observed information matrix when using the {EM}
  algorithm.
\newblock {\em Journal of the Royal Statistical Society. Series B
  (Methodological)}, {\bf 44}(2), 226--233.

\bibitem[McLachlan and Peel(1998)]{McLa:Peel:Robu:1998}
McLachlan, G. and Peel, D. (1998).
\newblock Robust cluster analysis via mixtures of multivariate
  $t$-distributions.
\newblock In A.~Amin, D.~Dori, P.~Pudil, and H.~Freeman, editors, {\em Advances
  in Pattern Recognition}, volume 1451 of {\em Lecture Notes in Computer
  Science}, pages 658--666. Springer Berlin - Heidelberg.

\bibitem[McLachlan and Basford(1988)]{McLa:Basf:mixt:1988}
McLachlan, G.~J. and Basford, K.~E. (1988).
\newblock {\em Mixture models: Inference and Applications to clustering}.
\newblock Marcel Dekker, New York.

\bibitem[McLachlan and Peel(2000)]{McLa:Peel:fini:2000}
McLachlan, G.~J. and Peel, D. (2000).
\newblock {\em Finite Mixture Models}.
\newblock John Wiley \& Sons, New York.

\bibitem[McNicholas(2010)]{McNi:Mode:2010}
McNicholas, P. (2010).
\newblock Model-based classification using latent gaussian mixture models.
\newblock {\em Journal of Statistical Planning and Inference}, {\bf 140}(5),
  1175--1181.

\bibitem[McNicholas and Murphy(2008)]{McNi:Murp:Pars:2008}
McNicholas, P. and Murphy, T. (2008).
\newblock Parsimonious {G}aussian mixture models.
\newblock {\em Statistics and Computing}, {\bf 18}(3), 285--296.

\bibitem[Peel and McLachlan(2000)]{Peel:McLa:2000}
Peel, D. and McLachlan, G. (2000).
\newblock {Robust mixture modelling using the $t$ distribution}.
\newblock {\em Statistics and Computing}, {\bf 10}(4), 339--348.

\bibitem[Rand(1971)]{Rand:Obje:1971}
Rand, W. (1971).
\newblock Objective criteria for the evaluation of clustering methods.
\newblock {\em Journal of the American Statistical Association}, {\bf 66}(336),
  846--850.

\bibitem[{\texttt{R} Development Core Team}(2011)]{R}
{\texttt{R} Development Core Team} (2011).
\newblock {\em \texttt{R}: A Language and Environment for Statistical
  Computing}.
\newblock \texttt{R} Foundation for Statistical Computing, Vienna, Austria.

\bibitem[Savageau and Loftus(1985)]{Boye:Sava:Plac:1985}
Savageau, D. and Loftus, G. (1985).
\newblock {\em Places Rated Almanac: your guide to finding the best places to
  live in America}.
\newblock Rand McNally \& Company, Chicago.

\bibitem[Schlattmann(2009)]{Schla:Medi:2009}
Schlattmann, P. (2009).
\newblock {\em Medical Applications of Finite Mixture Models}.
\newblock Springer-Verlag.

\bibitem[Schwarz(1978)]{Schw:Esti:1978}
Schwarz, G. (1978).
\newblock Estimating the dimension of a model.
\newblock {\em The Annals of Statistics}, {\bf 6}(2), 461--464.

\bibitem[Titterington {\em et~al.}(1985)]{Titt:Smit:Mako:stat:1985}
Titterington, D.~M., Smith, A. F.~M., and Makov, U.~E. (1985).
\newblock {\em Statistical Analysis of Finite Mixture Distributions}.
\newblock John Wiley \& Sons, New York.

\bibitem[Wedel and Kamakura(2001)]{Wede:Kama:Mark:2001}
Wedel, M. and Kamakura, W. (2001).
\newblock {\em Market segmentation: Conceptual and methodological foundations
  (2nd edition)}.
\newblock Kluwer Academic Publishers, Boston, MA, USA.

\end{thebibliography}

\end{document}